%% file: final.tex
\documentclass{article}
\usepackage{crc1-2e}
\usepackage{amssymb,amsfonts,latexsym,stmaryrd}
\usepackage[PostScript=dvips]{diagrams}
\usepackage{pstricks,pst-node,pst-tree}
\usepackage{proof}
\usepackage{QED}
\usepackage{epsfig}
\usepackage{bussproofs}

\newcommand{\vsa}{\vspace{.1in}}
\newcommand{\vsb}{\vspace{.2in}}
\newcommand{\vsn}{\vspace{-.1in}}

\newcommand{\name}[1]{\lceil #1 \rceil}
\newcommand{\coname}[1]{\lfloor #1 \rfloor}
\newcommand{\opp}{\mathsf{op}}

\input{macros}

\title{Temperley-Lieb Algebra: From  Knot Theory to Logic and Computation via  Quantum Mechanics}
\author{Samson Abramsky}
\date{12th March 2008}

\begin{document}

\maketitle

\section{Introduction}

Our aim in this paper is to trace some of the surprising and beautiful connections which are beginning to emerge between a number of apparently disparate topics.

\subsection{Knot Theory} 
Vaughan Jones' discovery of his new polynomial invariant of knots in 1984 \cite{Jon} triggered a spate of mathematical developments  relating knot theory, topological quantum field theory, and statistical physics \textit{inter alia} \cite{Wit,Kauffmanbook}. A central r\^ole, both in the initial work by Jones and in the subsequent developments, was played by what has come to be known as the \emph{Temperley-Lieb algebra}.\footnote{The original work of Temperley and Lieb \cite{TemLie} was in discrete lattice models of statistical physics. In finding exact solutions for a certain class of systems, they had identified the same relations which Jones, quite independently, found later in his work.} 

\subsection{Categorical Quantum Mechanics} 
Recently, motivated by the needs of Quantum Information and Computation, Abramsky and Coecke have recast \emph{the foundations of Quantum Mechanics itself}, in the more abstract language of category theory. The key contribution is the paper  \cite{AbrCoe2},  which  develops an axiomatic presentation of quantum mechanics in the general setting of \emph{strongly compact closed categories}, which  is adequate for the needs of Quantum Information and Computation. Moreover, this categorical axiomatics can be presented in terms of  a \emph{diagrammatic calculus} which is both intuitive and effective, and  can replace low-level computation with matrices by much more conceptual  reasoning. This diagrammatic calculus can be seen as a proof system for a logic \cite{AD04}, leading to a radically new perspective on what the right logical formulation for Quantum Mechanics should be.

This line of work has a direct connection to the Temperley-Lieb algebra, which can be put in a categorical framework, in which it can be described essentially as \emph{the free pivotal dagger category on one self-dual generator} \cite{FY}.\footnote{Strictly speaking, the full Temperley-Lieb category over a ring $R$ is the free $R$-linear enrichment of this free pivotal dagger category.} Here  \emph{pivotal dagger category} is a non-symmetric (``planar'') version of \emph{(strongly or dagger) compact closed category} --- the key notion in the Abramsky-Coecke axiomatics. 

\subsection{Logic and Computation}
The Temperley-Lieb algebra itself has some direct and striking connections to basic ideas in Logic and Computation, which offer an intriguing and promising bridge between these \textit{prima facie} very different areas.
We shall focus in particular on the following two topics:
\begin{itemize}
\item The Temperley-Lieb algebra has always hitherto been presented as a \emph{quotient} of some sort:  either algebraically by generators and relations as in Jones' original presentation \cite{Jon}, or as a diagram algebra modulo planar isotopy as in Kauffman's presentation \cite{Kau}. We shall use tools from \emph{Geometry of Interaction} \cite{GoI}, a dynamical interpretation of proofs under Cut Elimination developed as an off-shoot of Linear Logic \cite{LL}, to give a \emph{direct description} of the Temperley-Lieb category --- a \emph{fully abstract presentation}, in Computer Science terminology \cite{Fullabs}. This also brings something new to the Geometry of Interaction, since we are led to develop a planar version of it, and to verify that the interpretation of Cut-Elimination (the ``Execution Formula'' \cite{GoI}, or ``composition by  feedback'' \cite{NFGoI,Retrac}) preserves planarity.

\item We shall also show how the Temperley-Lieb algebra provides a natural setting in which computation can be performed diagrammatically as \emph{geometric simplification} --- ``yanking lines straight''. We shall introduce a ``planar $\lambda$-calculus'' for this purpose, and show how it can be interpreted in the Temperley-Lieb category.
\end{itemize}

\subsection{Outline of the Paper}
We briefly summarize the further contents of this paper.
In Section~2 we introduce the Temperley-Lieb algebras, emphasizing Kauffman's diagrammatic formulation. We also  briefly outline how the Temperley-Lieb algebra figures in the construction  of the Jones polynomial.
In Section~3 we describe the Temperley-Lieb category, which provides a more structured perspective on the Temperley-Lieb algebras. In Section~4, we discuss some features of this category, which have apparently not  been considered previously, namely a characterization of monics and epics, leading to results on image factorization and splitting of idempotents.
In Section~5, we briefly discuss the connections with the Abramsky-Coecke categorical formulation of Quantum Mechanics, and raise some issues and questions about the possible relationship betwen planar, braided and symmetric settings for Quantum Information and Computation.
In Section~6 we develop a planar version of Geometry of Interaction, and the direct ``fully abstract'' presentation of the Temperley-Lieb category. In Section~7 we discuss the planar $\lambda$-calculus and its interpretation in the Temperley-Lieb category.  We conclude in Section~8 with some further directions.

\textbf{Note to the Reader}
Since this paper aims at indicating cross-currents between several fields, it has been written in a somewhat expansive style, and an attempt has been made to explain the context of the various ideas we will discuss. We hope it will be accessible to readers with a variety of backgrounds.

\section{The Temperley-Lieb Algebra}
Our starting point is the Temperley-Lieb algebra, which has played a central role in the discovery by 
Vaughan Jones  of his new polynomial invariant of knots and links \cite{Jon}, and  in the subsequent developments over the past two decades relating knot theory, topological quantum field theory, and statistical physics \cite{Kauffmanbook}.

Jones' approach was algebraic: in his work, the Temperley-Lieb algebra was originally presented, rather forbiddingly, in terms of abstract generators and relations.
It was recast in beautifully elementary and conceptual terms by Louis Kauffman as a \emph{planar diagram algebra} \cite{Kau}. We begin with the algebraic presentation.

\subsection{Temperley-Lieb algebra: generators and relations}
We fix a ring $R$; in applications to knot polynomials, this is taken to be a ring of Laurent polynomials $\mathbb{C}[X, X^{-1}]$.
Given a choice of \emph{parameter} $\tau \in R$ and a \emph{dimension} $n \in \Nat$, we define the Temperley-Lieb algebra $\TLA{\tau}{n}$ to be the unital, associative $R$-linear algebra with generators
\[ U_{1} , \ldots , U_{n-1} \]
and relations
\[ \begin{array}{rclr}
U_{i}U_{j}U_{i} & = & U_{i} & \quad |i-j| = 1\\
U_{i}^{2} & = & \tau \cdot U_{i} & \\
U_{i}U_{j} & = & U_{j}U_{i} & |i-j| > 1
\end{array}
\]

Note that the only relations used in defining the algebra are multiplicative ones. This suggests that we can obtain the algebra $\TLA{\tau}{n}$ by presenting the multiplicative monoid $\TLM{n}$, and then obtaining $\TLA{\tau}{n}$ as the \emph{monoid algebra} of formal $R$-linear combinations $\sum_{i} r_{i} \cdot a_{i}$ over $\TLM{n}$, with the multiplication in $\TLA{\tau}{n}$ defined as the bilinear extension of the monoid multiplication in $\TLM{n}$:
\[ (\sum_{i} r_{i} \cdot a_{i})(\sum_{j} s_{j} \cdot b_{j}) = \sum_{i,j} (r_{i}s_{j})\cdot (a_{i}b_{j}) . \]
We define $\TLM{n}$ as the monoid with generators
\[ \delta , U_{1} , \ldots , U_{n-1} \]
and relations
\[ \begin{array}{rclr}
U_{i}U_{j}U_{i} & = & U_{i} & \quad |i-j| = 1\\
U_{i}^{2} & = & \delta U_{i} & \\
U_{i}U_{j} & = & U_{j}U_{i} & |i-j| > 1\\
\delta U_{i} & = & U_{i}\delta &
\end{array}
\]
We can then obtain $\TLA{\tau}{n}$ as the monoid algebra over $\TLM{n}$, subject to the identification 
\[ \delta = \tau \cdot 1 . \]

\subsection{Diagram Monoids}
These formal algebraic ideas are brought to vivid geometric life by Kauffman's interpretation of the monoids $\TLM{n}$ as \emph{diagram monoids}.

We start with two parallel rows of $n$ dots (geometrically, the dots are points in the plane).
The general form of an element of the  monoid  is obtained by ``joining up the dots'' pairwise in a smooth, planar fashion, where the arc connecting each pair of dots must lie within the rectangle framing the two parallel rows of dots. Such diagrams are identified up to planar isotopy, \ie continuous deformation within the portion of the plane bounded by the framing rectangle..

Thus the generators $U_{1}, \ldots , U_{n-1}$ can be drawn as follows:
\vspace{.1in}
\begin{center}
\input{./pictures/TLgens}
\end{center}
The generator $\delta$ corresponds to a loop $\bigcirc$ --- all such loops are identified up to isotopy.

We refer to arcs connecting dots in the top row as \emph{cups}, those connecting dots in the bottom row as \emph{caps}, and those connecting a dot in the top row to a dot in the bottom row as \emph{through lines}.

Multiplication $xy$ is defined by identifying the bottom row of $x$ with the top row of $y$, and composing paths. In general loops may be formed --- these are ``scalars'', which can float freely across these figures. The relations can be illustrated as follows:
\vsa
\begin{center}
\input{./pictures/TLrel12}

\end{center}
\vsa
\begin{center}
\input{./pictures/TLrel3}
\end{center}

\subsection{Expressiveness of the Generators}
The fact that all planar diagrams can be expressed as products of generators is not entirely obvious. For proofs, see \cite{Kau,Dos}. As an illustrative example, consider the planar diagrams in $\TLM{3}$. Apart from the generators $U_{1}, U_{2}$, and ignoring loops, there are three:
\begin{center}
\input{./pictures/TLM3.tex}
\end{center}
The first is the identity for the monoid; we refer to the other two as the \emph{left wave} and \emph{right wave} respectively. The left wave can be expressed as the product $U_{2}U_{1}$:
\begin{center}
\input{./pictures/wave.tex}
\end{center}
The right wave has a similar expression.

Once we are in dimension four or higher, we can have \emph{nested cups and caps}. These can be built using waves, as illustrated by the following:
\begin{center}
\input{./pictures/nest.tex}
\end{center}

\subsection{The Trace}
There is a natural  \emph{trace function} on the Temperley-Lieb algebra, which can be defined diagrammatically on $\TLM{n}$ by connecting each dot in the top row to the corresponding dot in the bottom row, using auxiliary cups and cups. This always yields a diagram isotopic to a number of loops --- hence to a \emph{scalar}, as expected. This trace can then be extended linearly to $\TLA{\tau}{n}$.

We illustrate this firstly by taking the trace of a wave---which is equal to a single loop:
\begin{center}
\input{./pictures/ear.tex}\\
The Ear is a Circle
\end{center}
Our second example illustrates the important general point that \emph{the trace of the identity in $\TLM{n}$ is $\delta^{n}$}:
\begin{center}
\input{./pictures/dimen.tex}
\end{center}

\subsection{The Connection to Knots}
How does this connect to knots? Again, a key conceptual insight is due to Kauffman, who saw how to recast the Jones polynomial in elementary combinatorial form in terms of his
\emph{bracket polynomial}.
The basic idea of the bracket polynomial is expressed by the following equation:
\begin{center}
\input{./pictures/Kbrac.tex}
\end{center}
Each over-crossing in a knot or link is evaluated to a weighted sum of the two possible planar smoothings. With suitable choices for the coefficients $A$ and $B$ (as Laurent polynomials), this is invariant under the second and third Reidemeister moves. With an ingenious choice of normalizing factor, it becomes invariant under the first Reidemeister move --- and yields the Jones polynomial!
What this means algebraically is that the braid group $\mathcal{B}_{n}$ has a representation in the Temperley-Lieb algebra $\TLA{\tau}{n}$ --- the above bracket evaluation shows how the generators $\beta_{i}$ of the braid group are mapped into the Temperley-Lieb algebra:
\[ \beta_{i} \;\; \mapsto \;\; A \cdot U_{i} + B \cdot 1 . \]
Every knot arises as the closure (\ie the diagrammatic trace) of a braid; the invariant arises by mapping the \emph{open braid} into the Temperley-Lieb algebra, and taking the trace there.

This is just the beginning of a huge swathe of further developments, including:
Topological Quantum Field Theories \cite{Wit},
Quantum Groups \cite{Kassel},
Quantum Statistical mechanics \cite{Kauffmanbook},
Diagram Algebras and Representation Theory \cite{HalvRam}, and more.

\section{The Temperley-Lieb Category}

We can expose more structure by gathering all the Temperley-Lieb algebras into a single category. We begin with the category $\DD$ which plays a similar role with respect to the diagram monoids $\TLM{n}$.

The objects of $\DD$ are the natural numbers. An arrow $\nn \rarr \mm$ is given by
\begin{itemize}
\item a number $k \in \Nat$ of loops
\item a diagram which joins the top row of $n$ dots and the bottom row of $m$ dots up pairwise, in the same smooth planar fashion as we have already specified for the diagram monoids. As before, diagrams are identified up to planar isotopy.
\end{itemize}
Composition of arrows $f : \nn \rarr \mm$  and $g : \mm \rarr \pp$ is defined by identifying the bottom row of $m$ dots for $f$ with the top row of $m$ dots for $g$, and composing paths. The loops in the resulting arrow are those of $f$ and of $g$, together with any formed by the process of composing paths.

Clearly we recover each $\TLM{n}$ as the endomorphism monoid $\DD (\nn,\nn)$.
Moreover, we can define the Temperley-Lieb category $\TL$ over a ring $R$ as the free $R$-linear category generated by $\DD$, with a construction which generalizes that of the monoid algebra: the objects of $\TL$ are the same as those of $\DD$, and arrows are $R$-linear combinations of arrows of $\DD$, with composition defined by bilinear extension from that in $\DD$:
\[ (\sum_{i} r_{i} \cdot g_{i}) \circ (\sum_{j} s_{j} \cdot f_{j}) = \sum_{i,j} (r_{i}s_{j})\cdot (g_{i} \circ f_{j}). \]
If we fix a parameter $\tau \in R$, then we obtain the category $\TL_{\tau}$ by the identification of the loop $\bigcirc$ in $\DD$ with the scalar $\tau$ in $\TL$.\footnote{The full justification of this step requires the identification of $\DD$ as a free pivotal category, as discussed below.} We then recover the Temperley-Lieb algebras as
\[ \TLA{\tau}{n} = \TL_{\tau}(\nn,\nn) . \]
New possibilities also arise in $\DD$. In particular, we get the \emph{pure cap}
\begin{center}
\input{./pictures/purecap.tex}
\end{center}
as (the unique) arrow $\zz \rarr \mathbf{2}$, and similarly the \emph{pure cup} as the unique arrow $\mathbf{2} \rarr \zz$. More generally, for each $n$ we have
arrows $\eta_{\nn} : \zz \rarr \nn + \nn$, and
$\epsilon_{\nn} : \nn + \nn \rarr \zz$:
\begin{center}
\begin{minipage}[t]{1in}
\input{./pictures/unitn.tex}
\end{minipage}
\ $\qquad\qquad\qquad$ \
\raisebox{-2ex}{
\begin{minipage}[t]{1in}
\input{./pictures/counitn.tex}
\end{minipage}
}
\end{center}
We refer to the arrows $\eta_{\nn}$ as \emph{units}, and the arrows $\epsilon_{\nn}$ as \emph{counits}.

The category $\DD$ has a natural \emph{strict monoidal structure}. On objects, we define $\nn \otimes \mm = \nn + \mm$, with unit given by $I = \zz$. The tensor product of morphisms
\[  \frac{f : \nn \rarr \mm \quad g : \pp \rarr \qq}{f \otimes g : \nn + \pp \rarr \pp + \qq} \]
is given by juxtaposition of diagrams in the evident fashion, with (multiset) union of loops.
Thus we can write the units and counits as arrows
\[ \eta_{\nn} : I \rarr \nn \otimes \nn, \qquad \epsilon_{\nn} : \nn \otimes \nn \rarr I . \]
These units and counits satisfy important identities, which we illustrate diagrammatically
\begin{center}
\input{./pictures/cupcapeqn.tex}
\end{center}
and write algebraically as
\begin{equation}
\label{triangeq}
(\epsilon_{\nn} \otimes 1_{\nn}) \circ (1_{\nn} \otimes \eta_{\nn}) = 1_{\nn} = (1_{\nn} \otimes \epsilon_{\nn}) \circ (\eta_{\nn} \otimes 1_{\nn}) . 
\end{equation}

\subsection{Pivotal Categories}

From these observations, we see that $\DD$ is a \emph{strict pivotal category} \cite{FY}, in which the duality on objects is trivial: $A = A^{*}$.
We recall that a strict pivotal category is a strict monoidal category $(\CC, \otimes, I)$ with an assignment $A \mapsto A^{*}$ on objects satisfying
\[ A^{**} = A, \qquad (A \otimes B)^{*} = B^{*} \otimes A^{*}, \qquad I^{*} = I , \]
and for each object $A$, arrows 
\[  \eta_{A} : I \rarr A^{*} \otimes A, \qquad \epsilon_{A} : A \otimes A^{*} \rarr I  \]
satisfying the triangular identities:
\begin{equation}
\label{triangeq2}
(\epsilon_{A} \otimes 1_{A}) \circ (1_{A} \otimes \eta_{A}) = 1_{A}, \qquad
(1_{A^{*}} \otimes \epsilon_{A}) \circ (\eta_{A} \otimes 1_{A^{*}})  = 1_{A^{*}}. 
\end{equation}
In addition, the following coherence equations are required to hold:
\[  \eta_{I} = 1_{I}, \qquad \eta_{A \otimes B} = (1_{B^{*}} \otimes \eta_{A} \otimes 1_{B}) \circ \eta_{B} , \]
and, for $f : A \rarr B$:
\[  
\begin{diagram}[3em]
& & A^{*} \otimes A \otimes B^{*} & \rTo^{1 \otimes f \otimes 1} &  A^{*} \otimes B \otimes B^{*} & & \\
& \ruTo^{\eta_{A} \otimes 1} &  & & & \rdTo^{1 \otimes \epsilon_{B}} & \\
B^{*} & & & & & & A^{*} \\
& \rdTo_{1 \otimes \eta_{A^{*}}} &  & & & \ruTo_{\epsilon_{B^{*}}\otimes 1} & \\
& &  B^{*}  \otimes A \otimes A^{*} & \rTo_{1 \otimes f \otimes 1} &  B^{*} \otimes B \otimes A^{*} & & 
\end{diagram}
\]
This last equation is illustrated diagrammatically by
\begin{center}
\input{./pictures/stareqn.tex}
\end{center}
We extend $()^{*}$ to a contravariant involutive functor:
\[ \frac{f : A \rarr B}{f^{*} : B^{*} \rarr A^{*}} \qquad f^{*} = (1 \otimes \epsilon_{A}) \circ (1 \otimes f \otimes 1) \circ (\eta_{A} \otimes 1) \]
which indeed satisfies
\[ 1^{*} = 1, \qquad (g \circ f)^{*} = f^{*} \circ g^{*}, \qquad f^{**} = f , \]
the last equation being illustrated diagrammatically by

\vspace{.1in}
\begin{center}
\input{./pictures/dualinv.tex}
\end{center}

\vspace{.1in}
These axioms have powerful consequences. In particular, $\CC$ is \emph{monoidal closed}, with internal hom given by $A^{*} \otimes B$, and the adjunction:
\[ \CC (A \otimes B, C) \simeq \CC(B, A^{*} \otimes C) :: f \mapsto (1 \otimes f) \circ (\eta_{A} \otimes 1) . \]
This means that a restricted form of $\lambda$-calculus can interpreted in such categories --- a point we shall return to in Section~7.

A \emph{trace function} can be defined in pivotal categories, which takes an endmorphism $f : A \rarr A$ to a \emph{scalar} in $\CC (I, I)$:
\[ \Trace{f} = \epsilon_{A} \circ (f \otimes 1) \circ \eta_{A^{*}} . \]
 It satisfies:
\[ \Trace{g \circ f} = \Trace{f \circ g} . \]
In $\DD$, this definition yields exactly the diagrammatic trace  we discussed previously.

We have the following important characterization of the diagrammatic category $\DD$:
\begin{proposition}
\label{freepivot}
$\DD$ is the free pivotal category over one self-dual generator; that is, freely generated over the one-object one-arrow category, with object $A$ say, subject to the equation $A = A^{*}$.
\end{proposition}
This was mentioned (although not proved) in \cite{FY}; see also \cite{Dos}.
The methods in \cite{Abr05} can be adapted to prove this result, using the ideas we shall develop in Section~6.

The idea of ``identifying the loop with the scalar $\tau$'' in passing from $\DD$ to the full Temperley-Lieb category $\TL_{\tau}$ can be made precise using the construction given in \cite{Abr05} of gluing a specified ring $R$ of scalars onto a free compact closed category, along a given map from the loops in the generating category to $R$. In this case, there is a single loop in the generating category, and we send it to $\tau$.

\subsection{Pivotal Dagger Categories}
We now mention a strengthening of the axioms for pivotal categories, corresponding to the notion of \emph{strongly compact closed} or \emph{dagger compact closed} category which has proved to be important in the categorical approach to Quantum Mechanics \cite{AbrCoe2,AbrCoe3}. Again we give the strict version for simplicity. We assume that the strict monoidal category $(\CC, \otimes, I)$ comes equipped with an identity-on-objects, contravariant involutive functor $()^{\dagger}$ such that $\epsilon_{A} = \eta_{A^{*}}^{\dagger}$.
The idea is that $f^{\dagger}$ abstracts from the \emph{adjoint} of a linear map, and allows the extra structure arising from the use of \emph{complex} Hilbert spaces in Quantum Mechanics to be expressed in the abstract setting.

Note that there is a clear diagrammatic distinction between the dual $f^{*}$ and the adjoint $f^{\dagger}$. The dual corresponds to $180^{\circ}$ rotation in the plane:

\vsb
\begin{center}
\raisebox{-5ex}{
\input{./pictures/dualbox.tex}
}
\end{center}

\vsb
\noindent while the adjoint is reflection in the $x$-axis:

\vsb
\begin{center}
\raisebox{-5ex}{
\input{./pictures/daggerbox.tex}
}
\end{center}

\vsb
\noindent For example in $\DD$, if we consider the left and right wave morphisms $L$ and $R$:
\begin{center}
\input{./pictures/LRwave.tex}
\end{center}
then we have 
\[ L^{*} = L, \quad  L^{\dagger} = R, \quad R^{*} = R, \quad R^{\dagger} = L . \]

\noindent Using the adjoint, we can define a \emph{covariant} functor
\[ \frac{f : A \rarr B}{f_{*} : A^{*} \rarr B^{*}} \qquad f \mapsto f^{*\dagger} . \]
We have
\[ (f^{*})_{*} = f^{\dagger} = (f_{*})^{*} . \]
In terms of complex matrices, $f^{*}$ is transpose, while $f_{*}$ is complex conjugation. Diagrammatically, $f_{*}$ is ``reflection in the $y$-axis''. 
\vsa
\begin{center}
\input{./pictures/fourex.tex}
\end{center}

\vsb
We have the following refinement of Proposition~\ref{freepivot}, by similar methods to those used for free strongly compact closed categories in \cite{Abr05}.
\begin{proposition}
\label{freepivotdagprop}
$\DD$ is the free pivotal dagger category over one self-dual generator.
\end{proposition}

\section{Factorization and Idempotents}
We now consider some structural properties of the category $\DD$ which we have not found elsewhere in the literature.\footnote{The idea of considering these properties arose from a discussion with Louis Kauffman, who showed the author a direct diagrammatic characterization of idempotents in $\DD$, which has subsequently appeared in \cite{KauII}.}

We begin with a pleasingly simple diagrammatic characterization of \emph{monics} and \emph{epics} in $\DD$.

\begin{proposition}
\label{emchar}
An arrow in $\DD$ is monic iff it has no cups; it is epic iff it has no caps.
\end{proposition}
\begin{proof}
Suppose that $f : \nn \rarr \mm$ has no cups. Thus all dots in $\nn$ are connected by through lines to dots in $\mm$. Now consider a composition $f \circ g$. No loops can be formed by this composition; hence we can recover $g$ from $f \circ g$ by erasing the caps of $f$. Moreover, the number of loops in $f \circ g$ will simply be the sum of the loops in $f$ and $g$, so we can recover the loops of $g$ by subtracting the loops of $f$ from the composition. It follows that
\[ f \circ g = f \circ h \;\; \Longrightarrow \;\; g = h, \]
\ie that $f$ is monic, as required.

For the converse, suppose that $f$ has a cup, which we can assume to be connecting dots $i$ and $i+1$ in the top row. (Note that if $i < j$ are connected by a cup, then by planarity, every $k$ with $i < k < j$ must also be connected in a cup to some $l$ with $i < l < j$.) Then $f \circ \delta\cdot1 = f \circ (1 \otimes U_{i} \otimes 1)$, so $f$ is not monic. Diagrammatically, this says that we can either form a loop using the cup of $f$, or simply add a loop which is attached to an identity morphism.

The characterization of epics is entirely similar.
\end{proof}

This immediately yields an ``image factorization'' structure for $\DD$.
\begin{proposition}
\label{emfact}
Every arrow in $\DD$ has an epi-mono factorization. 
\end{proposition}
\begin{proof}
Given an arrow $f :  \nn \rarr \mm$, suppose it has $p$ cups and $q$ caps. Then we obtain arrows
$e : \nn \rarr (\mm - 2\qq)$ by erasing the caps, and $m : (\nn - 2\pp) \rarr \mm$ by erasing the cups. By Proposition~\ref{emchar}, $e$ is epic and $m$ monic. Moreover, the number of dots in the top and bottom rows connected by through lines must be the same. Hence 
\[ (\mm - 2\qq) = \kk = (\nn - 2\pp) , \]
and we can compose $e$ and $m$ to recover $f$. Note that by planarity, once we have assigned cups and caps, there is no choice about the correspondence between top and bottom row dots by through lines.

This factorization is ``essentially'' unique. However, we are free to split the $l$ loops of $f$ between $e$ and $m$ in any way we wish, so there is a distinct factorization $\delta^{a}\cdot m \circ \delta^{b}\cdot e$ for all $a, b \in \Nat$ with $a + b = l$.
\end{proof}
We illustrate the epi-mono factorization for the left wave:
\begin{center}
\input{./pictures/emfact.tex}
\end{center}

We recall that an \emph{idempotent} in a category is an arrow $i : A \rarr A$ such that $i^{2} = i$. We say that an idempotent $i$ \emph{splits} if there are arrows $r : A \rarr B$ and $s : B \rarr A$ such that
\[ i = s \circ r, \qquad r \circ s = 1_{B}. \]

\begin{proposition}
All idempotents split in $\DD$.
\end{proposition}
\begin{proof}
Let $i : \nn \rarr \nn$ be an idempotent in $\DD$. By Proposition~\ref{emfact}, $i = m \circ e$, where $e : \nn \rarr \kk$ is epic and $m : \kk \rarr \nn$ is monic. Now 
\[ m \circ e \circ m \circ e = m \circ e . \]
Since $m$ is monic, this implies that $e \circ m \circ e = e = 1 \circ e$. Since $e$ is epic, this implies that $e \circ m = 1$.
\end{proof}

\section{Categorical Quantum Mechanics}
We now relate our discussion to the Abramsky-Coecke programme of Categorical Quantum Mechanics.

This approach is very different to previous work on the Computer Science side of this interdisciplinary area, which has focussed on quantum algorithms and complexity. The focus has rather been on developing \emph{high-level methods} for Quantum Information and Computation (QIC)---languages, logics, calculi, type systems etc.---analogous to those which have proved so essential in classical computing \cite{HLM}. This has led to nothing less than a recasting of \emph{the foundations of Quantum Mechanics itself}, in the more abstract language of category theory. The key contribution is the paper with Coecke \cite{AbrCoe2}, in which we develop an axiomatic presentation of quantum mechanics in the general setting of \emph{strongly compact closed categories}, which  is adequate for all the needs of QIC.

Specifically, we show that we can recover the key quantum mechanical notions of \em  inner-product, unitarity, full and partial trace, Hilbert-Schmidt inner-product and map-state duality, projection,  positivity, measurement\em, and \em Born rule \em (which provides the quantum \em probabilities\em),  axiomatically at this high level of abstraction and generality. Moreover,  we can  derive the correctness of protocols such as quantum teleportation, entanglement swapping and logic-gate teleportation  \cite{BBC,Gottesman,Swap} in a transparent and very conceptual fashion.  Also, while at this level of abstraction there is no underlying field of complex numbers, there \emph{is} still an intrinsic notion of `scalar', and we can still make sense of \em dual vs.~adjoint \em \cite{AbrCoe2,AbrCoe3}, and  \emph{global phase and elimination thereof} \cite{deLL}. Peter Selinger recovered \em mixed state, complete positivity \em and  \em Jamiolkowski map-state duality \em \cite{Selinger}.
Recently,  in collaboration with Dusko Pavlovic and Eric Paquette,  \em decoherence, generalized measurements \em and \em Naimark's theorem \em have been recovered \cite{CoePav,Paquette}.

Moreover, this formalism has two important additional features. Firstly,  it goes \emph{beyond} the standard Hilbert-space formalism, in that it is able to capture classical as well as quantum information flows, and the interaction between them, \emph{within the formalism}. For example, we can capture the idea that the result of a measurement is used to determine a further stage of quantum evolution, as e.g.~in the teleportation protocol \cite{BBC}, where a unitary correction must be performed after a measurement; or also in measurement-based quantum computation \cite{Briegel,Briegel2}. Secondly, this categorical axiomatics can be presented in terms of  a \emph{diagrammatic calculus} which is extremely intuitive, and potentially can replace  low-level computation with matrices by much more conceptual --- and  automatable --- reasoning. Moreover, this diagrammatic calculus can be seen as a proof system for a logic, leading to a radically new perspective on what the right logical formulation for Quantum Mechanics should be. This latter topic is initiated in \cite{AD04}, and developed further in the forthcoming thesis of Ross Duncan.

\subsection{Outline of the approach}
We now give some further details of the approach.
The general setting is that of \emph{strongly (or dagger) compact closed categories}, which are the symmetric version of the pivotal dagger categories we encountered in Section~3. Thus, in addition to the structure mentioned there, we have a symmetry natural isomorphism
\[ \sigma_{A,B} : A \otimes B \simeq B \otimes A . \]
See \cite{AbrCoe3} for an extended discussion.
An important feature of the Abramsky-Coecke approach is the use of an intuitive \emph{graphical calculus}, which is essentially the diagrammatic formalism we have seen in the Temperley-Lieb setting, extended with more general basic types and arrows. The key point is that this formalism admits a very direct \emph{physical interpretation} in Quantum Mechanics.

In the graphical calculus we depict physical processes by boxes, and we label the inputs and outputs of these boxes by \em types \em which indicate the kind of system on which these boxes act,  e.g. one qubit, several qubits, classical data, etc. Sequential composition (in time) is depicted by connecting matching outputs and inputs by wires, and parallel composition (tensor) by locating entities side by side e.g.
\[
1_A:A\to A\  \quad f:A\to B\  \quad g\circ f\ \ \quad 1_A\otimes 1_B\  \quad f\otimes 1_C\  \quad f\otimes g\  \quad (f\otimes g)\circ h
\]
for $g:B\to C$ and $h:E\to A\otimes B$ are respectively depicted as:

\medskip\noindent\begin{minipage}[b]{1\linewidth}
\centering{\epsfig{figure=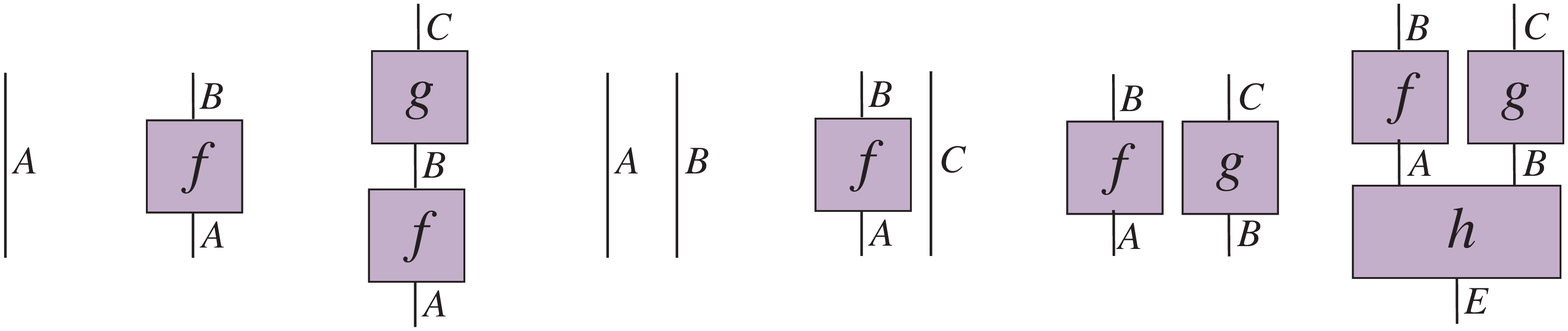,width=300pt}}      
\end{minipage}

\noindent
(The convention in these diagrams is that the `upward' vertical direction represents progress of time.)
A special role is played by boxes with either no input or no output,  called \em states \em and \em costates \em respectively (cf.~Dirac's kets and bras  \cite{Dirac}) which we depict by triangles. Finally, we also need to consider diamonds which arise by post-composing a state with a matching costate (cf.~inner-product or Dirac's bra-ket):

\smallskip\noindent\begin{minipage}[b]{1\linewidth}
\centering{\epsfig{figure=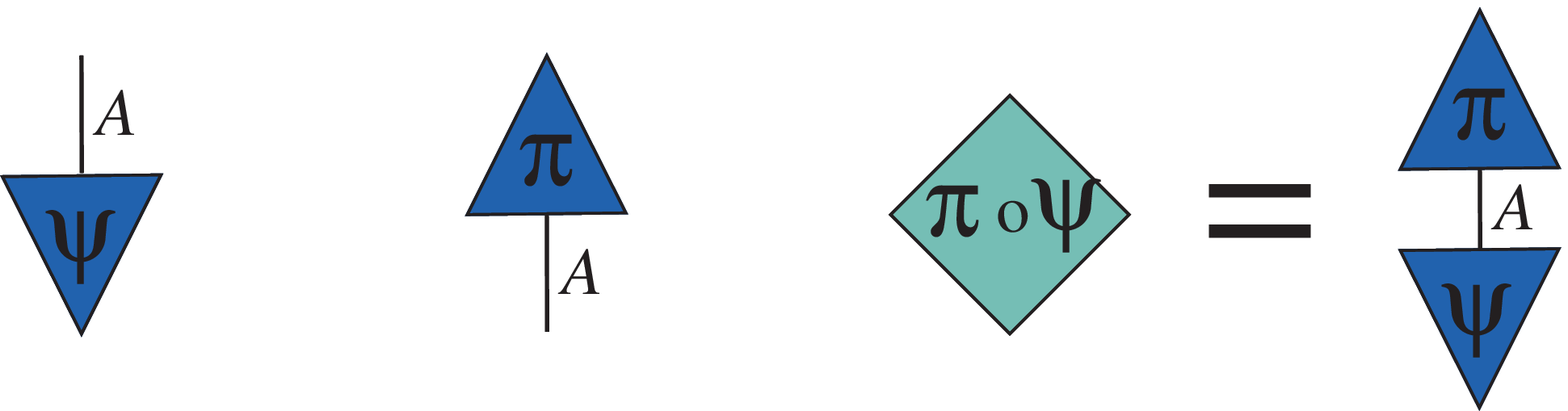,width=180pt}}     
\end{minipage}

\noindent
that is, algebraically,
\[
\psi:\II\to A\quad\qquad \pi:A\to \II \quad\qquad \pi\circ\psi:\II\to\II
\]
where $\II$ is the \emph{tensor unit}:  $A\otimes\II\simeq A\simeq \II\otimes A$.
Extra structure is represented by (i) assigning a direction to the wires, where reversal of this direction is denoted by $A\mapsto A^*$, (ii) allowing reversal of boxes (cf.~the \em adjoint \em for vector spaces), and, (iii) assuming that for each type $A$ there exists  a special bipartite \em Bell-state \em and its adjoint \em Bell-costate\em:

\smallskip\noindent\begin{minipage}[b]{1\linewidth}
\centering{\epsfig{figure=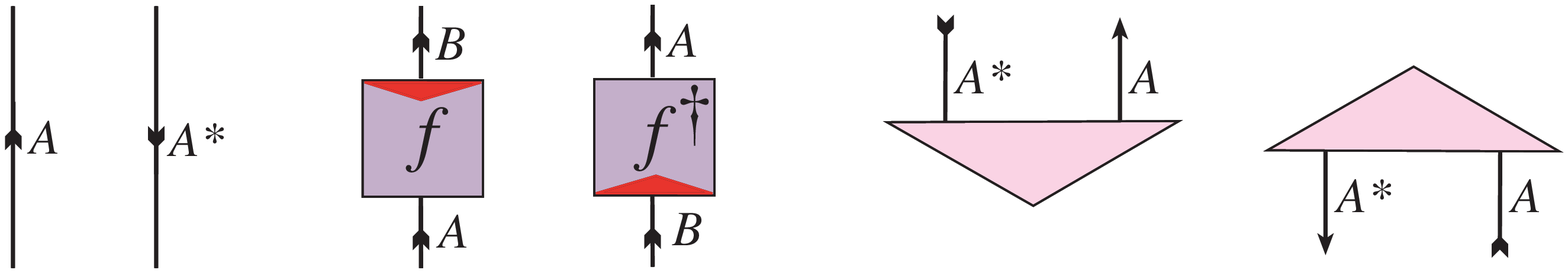,width=285pt}}     
\end{minipage}

\noindent
that is, algebraically,
\[
A\quad\ A^*\quad\ f:A\to B\quad f^\dagger:B\to A\quad\  \eta_A:\II\to A^*\otimes A
\quad\  \eta^\dagger_A:A^*\otimes A\to\II.
\]
Hence, bras and kets are adjoint and the inner product has the form ${(-)^\dagger\circ(-)}$ on states.  Essentially the sole \em axiom \em we impose is:

\smallskip\noindent
\begin{minipage}[b]{1\linewidth}
\centering{\epsfig{figure=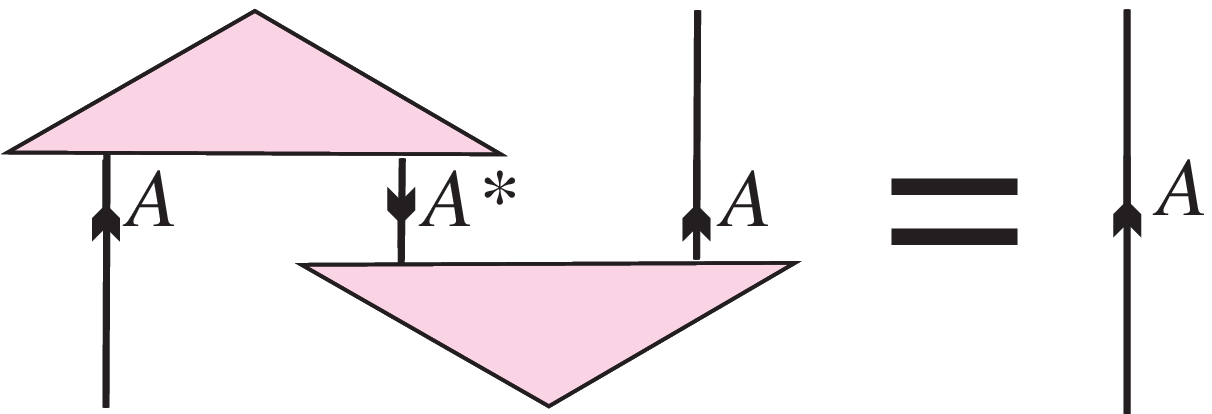,width=125pt}}      
\end{minipage}

\noindent
that is, algebraically,
\[
(\eta^\dagger_{A^*}\otimes 1_A)\circ(1_A\otimes \eta_A)=1_A\,.
\]
If we extend the graphical notation of Bell-(co)states to: 

\smallskip\noindent
\begin{minipage}[b]{1\linewidth}
\centering{\epsfig{figure=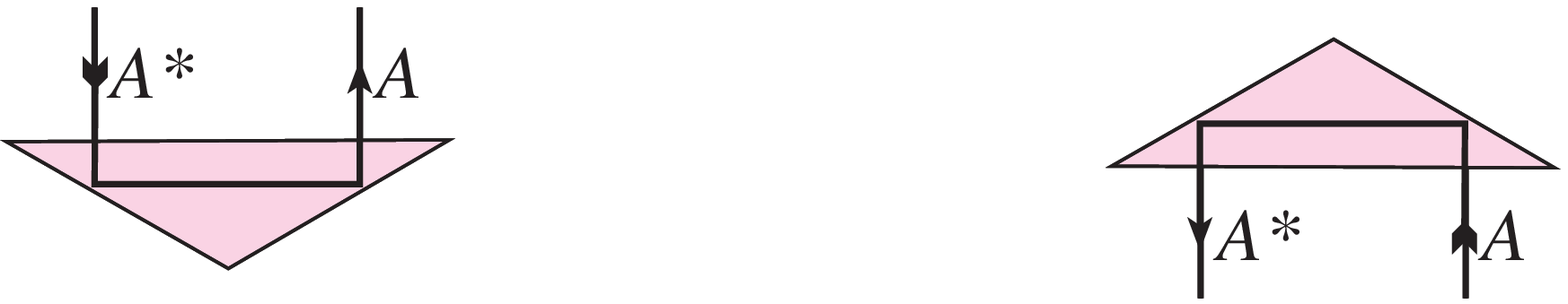,width=200pt}\quad\quad\ }     
\end{minipage}

\noindent
we obtain a clear graphical interpretation for the axiom:
\begin{center}
\begin{minipage}[b]{1\linewidth}
\centering{\fbox{\qquad\ \epsfig{figure=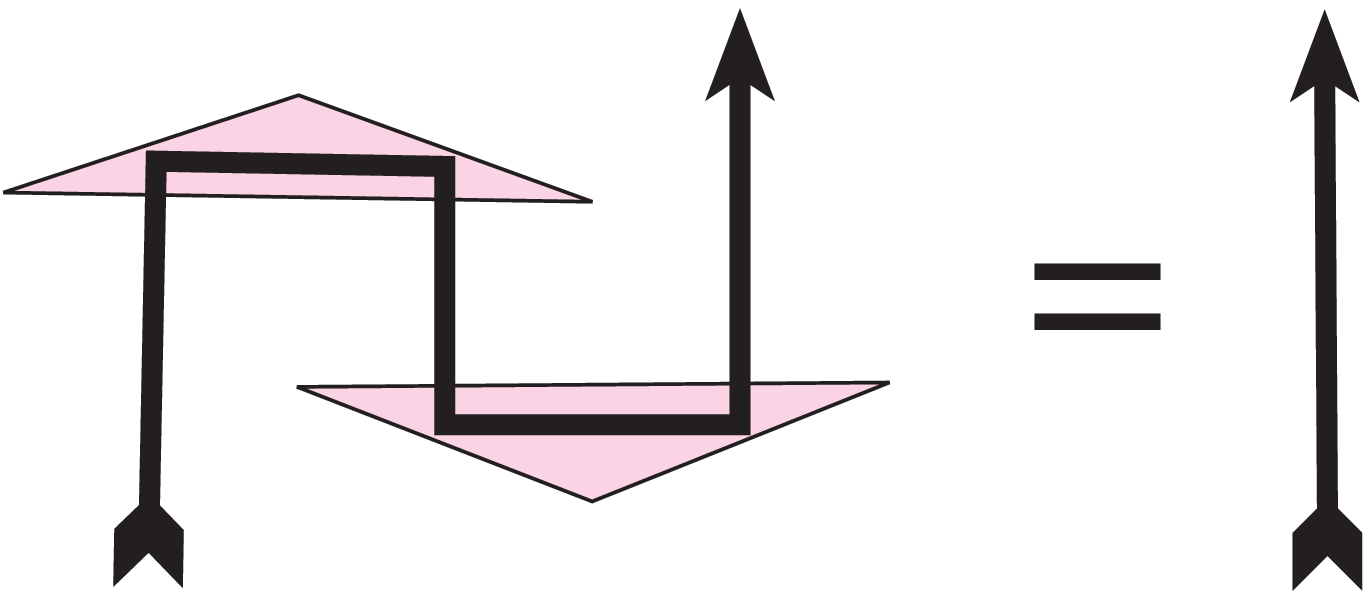,width=130pt}\quad\hfill{\bf (1)}}}
\end{minipage}
\end{center}

\medskip\noindent
which now tells us that we are allowed to \em yank  the black line straight \em:

\smallskip\noindent
\begin{minipage}[b]{1\linewidth}
\centering{\epsfig{figure=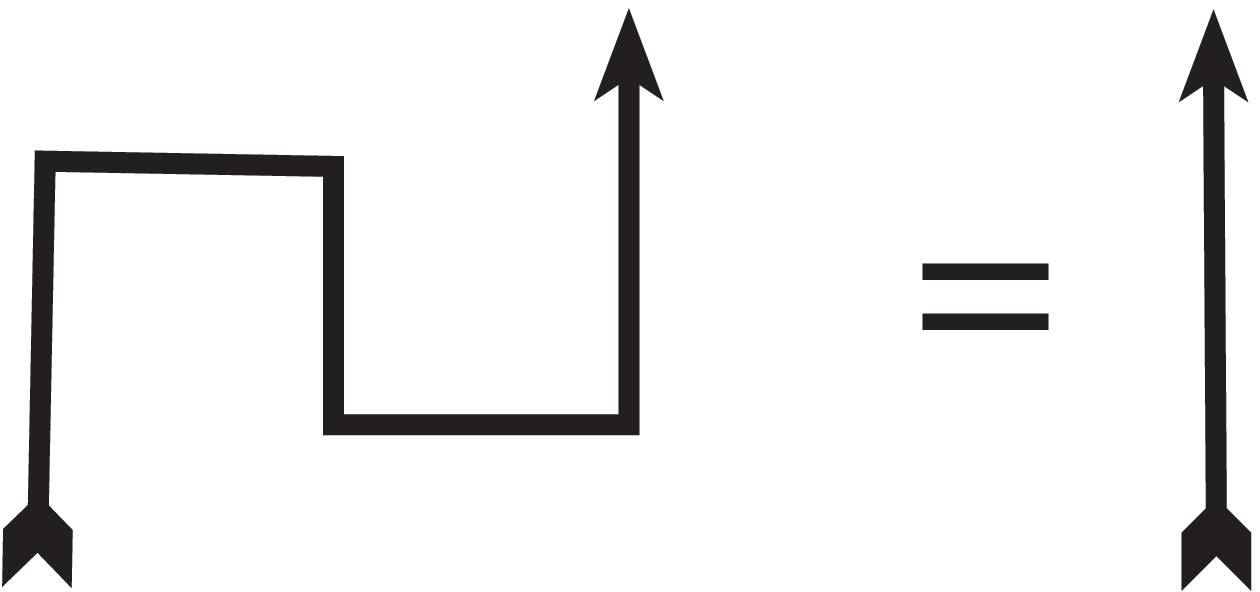,width=130pt}}     
\end{minipage}

\noindent
This equation and its diagrammatic counterpart should of course be compared to equation (\ref{triangeq2}), and equation (\ref{triangeq}) and its accompanying diagram, in Section~3--- they are one and the same, subject to minor differences in diagrammatic conventions.

This intuitive graphical calculus is an important benefit of the categorical axiomatics.  Other advantages can be found in \cite{AbrCoe2,HLM}.

\subsection{Quantum non-logic vs.~quantum hyper-logic}

The term \em quantum logic \em is usually understood in connection with the 1936 Birkhoff-von Neumann proposal \cite{BvN,Redei} to consider the (closed) linear subspaces of a Hilbert space  ordered by inclusion as the formal expression of the logical distinction between quantum and classical physics.  While in classical logic we have deduction, the linear subspaces of a Hilbert space form a non-distributive lattice and hence there is no obvious notion of implication or deduction.  Quantum logic was therefore always seen as logically  very weak, or even as a non-logic.  In addition, it has never given a satisfactory account of compound systems and entanglement. 

On the other hand, \em compact closed logic \em in a sense goes beyond ordinary logic in the principles it admits. Indeed, while in ordinary categorical logic ``logical deduction'' implies that \em morphisms internalize  as elements \em (which above we referred to above as \emph{states}) i.e.
\[
B\rTo^{f}C
\ \ \ \ \stackrel{\simeq}{\longleftrightarrow}\ \ \ \
I\rTo^{\name{f}} B\!\Rightarrow\! C
\]
(where $I$ is the tensor unit), in \em compact closed logic \em they internalize \emph{both} as states \emph{and} as costates, \ie
\[
A\otimes B^*\!\rTo^{\coname{f}} I
\ \ \ \stackrel{\simeq}{\longleftrightarrow}\ \ \
A\rTo^{f}B
\ \ \ \stackrel{\simeq}{\longleftrightarrow}\ \ \
I \rTo^{\name{f}} A^*\!\otimes B
\]
where  we introduce the following notation:
\[
\uu f\uuu= (1_{A^*}\otimes f)\circ\eta_A: I \to A^*\otimes B
\quad \dd f\ddd=\epsilon_{B}\circ(f\otimes 1_{B^*}):A\otimes B^*\to I .
\]
It is exactly this dual internalization which allows  the \em straightening axiom  \em in picture {\bf(1)} to be expressed. In the graphical calculus this  is witnessed by the fact that we can  define both a state and a costate

\bigskip\noindent
\begin{minipage}[b]{1\linewidth}
\centering{\fbox{\epsfig{figure=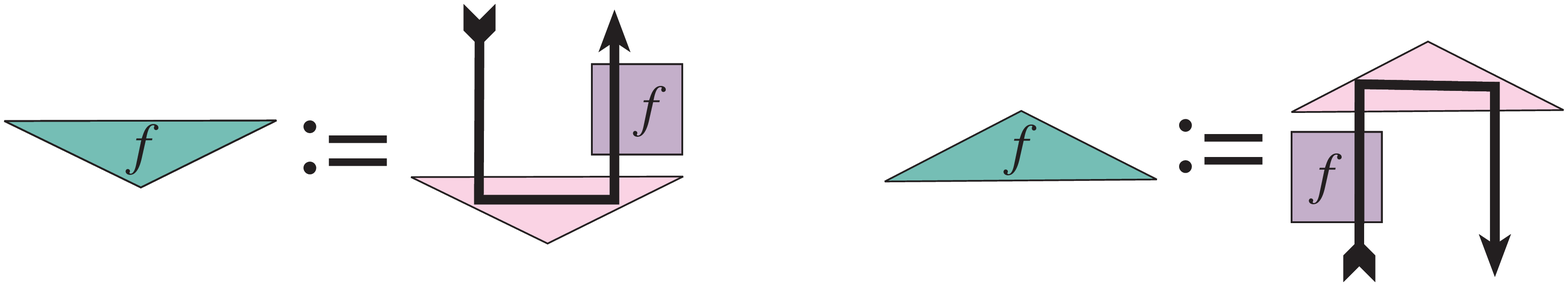,width=285pt}\hfill{\bf (2)}}}     
\end{minipage}

\bigskip\noindent
for each operation $f$. Physically, costates form the (destructive parts of) \emph{projectors}, i.e. branches of projective measurements.

\subsubsection{Compositionality.} The semantics is obviously compositional,  both with respect to sequential composition of operations and parallel composition of types and operations, allowing the description of systems to be built up from smaller components.
But we also have something more specific in mind: a form of compositionality with direct applications to the analysis of compound entangled systems. Since we have:

\noindent
\begin{minipage}[b]{1\linewidth}
\centering{\epsfig{figure=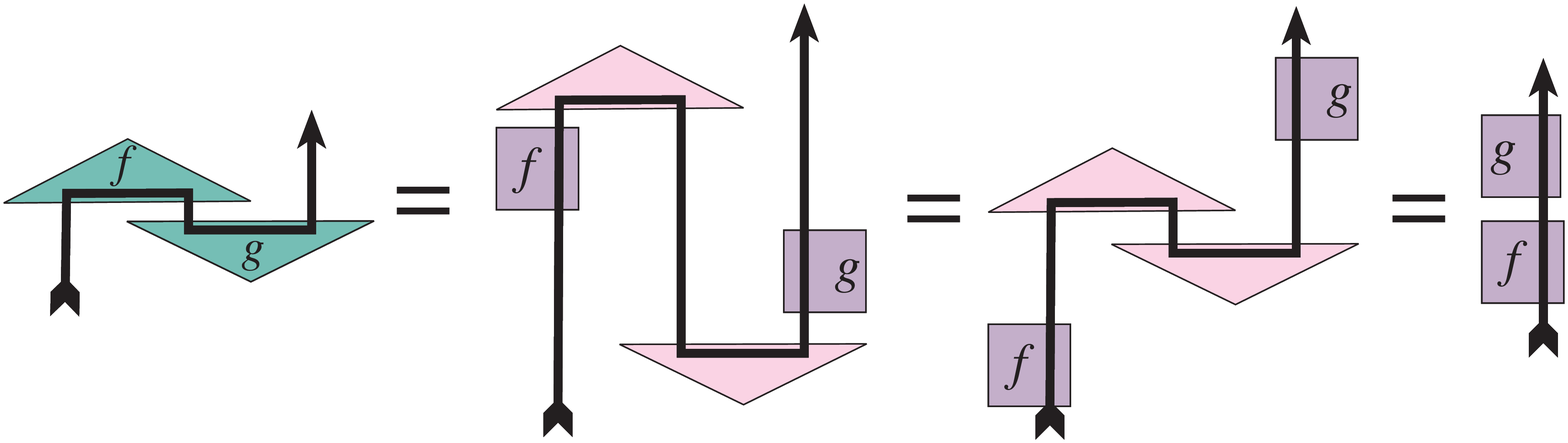,width=300pt}}     
\end{minipage}

\noindent
we obtain:

\medskip\noindent
\begin{minipage}[b]{1\linewidth}
\centering{\fbox{\qquad\epsfig{figure=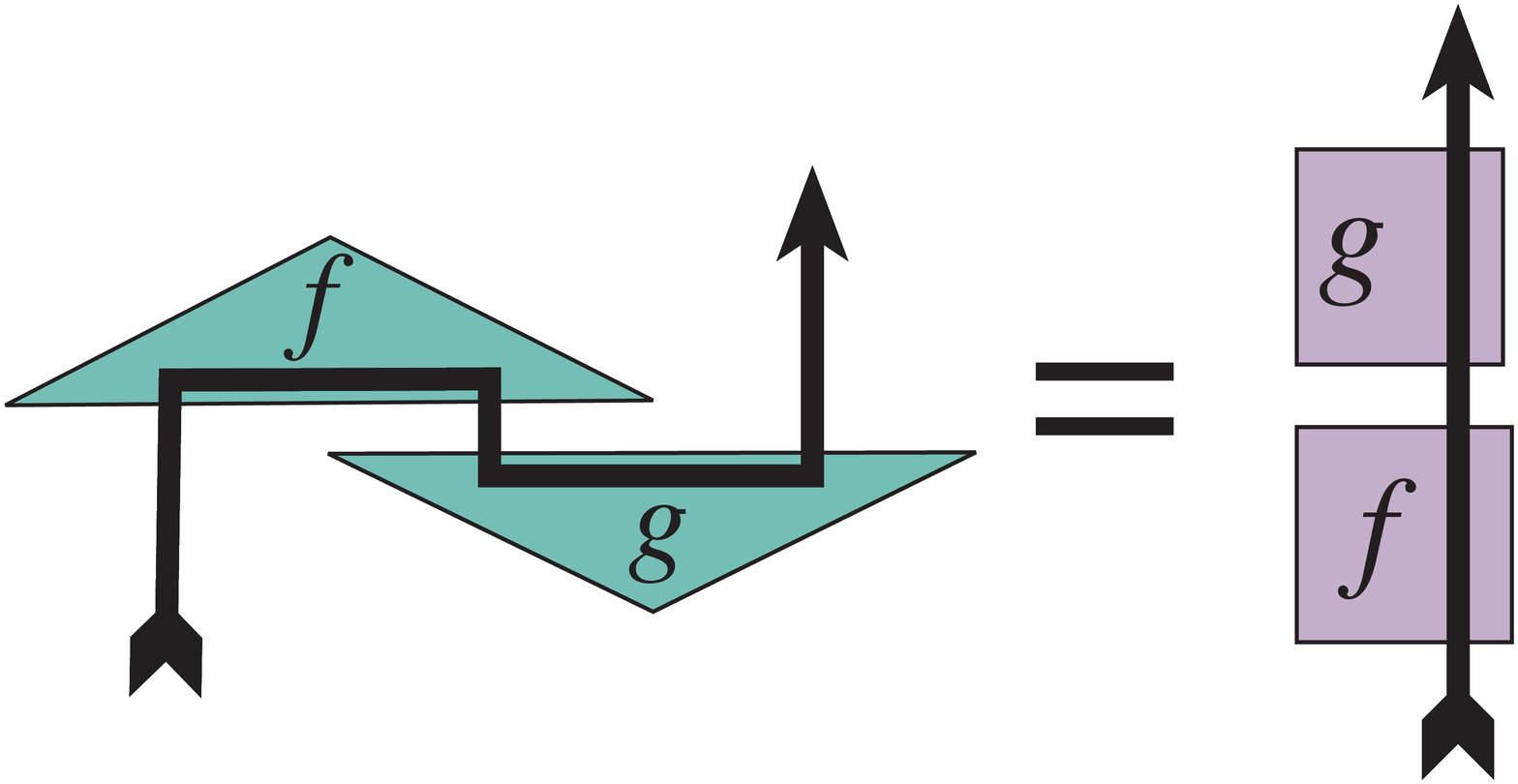,width=140pt}\qquad\hfill{\bf (3)}}}     
\medskip\end{minipage}

\noindent
i.e.~composition of operations can be \em internalized \em in the behavior of entangled states and costates. Note in particular the interesting phenomenon of ``apparant reversal of the causal order'' which is the source of many quite mystical interpretations of quantum teleportation in terms of ``traveling backward in time'' --- cf.~\cite{Laflamme}.  Indeed, while on the left, physically, we first prepare the state labeled $g$ and then apply the costate labeled $f$, the global effect is {\em as if} we first applied $f$ itself first, and only then $g$.

\subsubsection{Derivation of quantum teleportation.} This is the most basic application of compositionality in action.  Immediately from picture {\bf(1)} we can read the quantum mechanical potential for teleportation:

\bigskip\noindent
\begin{minipage}[b]{1\linewidth}
\centering{$\!\!$\epsfig{figure=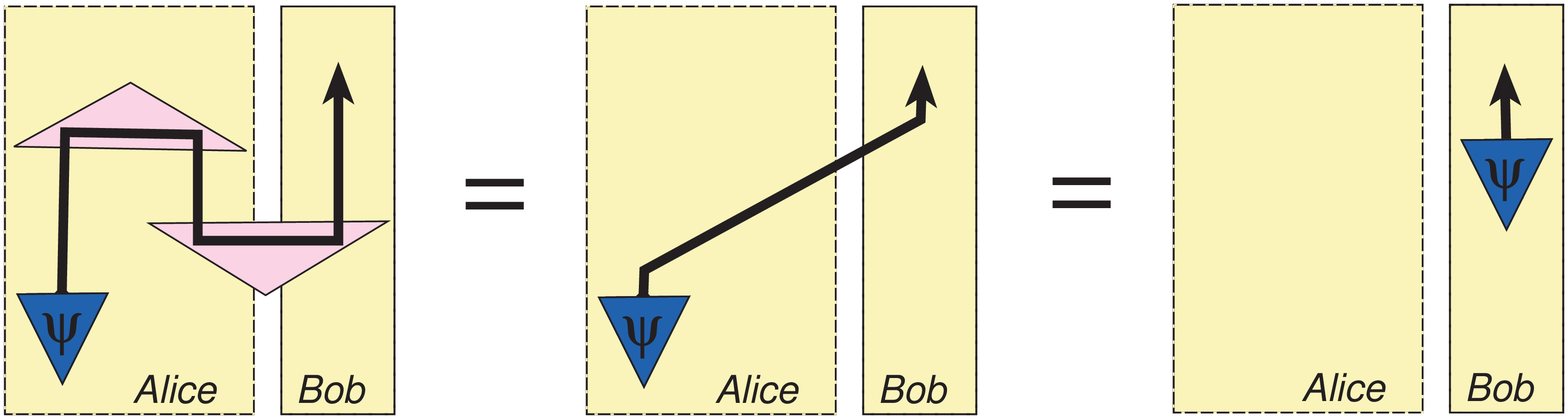,width=300pt}}     
\end{minipage}

\smallskip\noindent
This is not quite the whole story, because of  the non-deterministic nature of measurements.
But it suffices to introduce a unitary correction. Using picture {\bf(3)}
the full description of teleportation becomes:

\medskip\noindent
\begin{minipage}[b]{1\linewidth}
\centering{\epsfig{figure=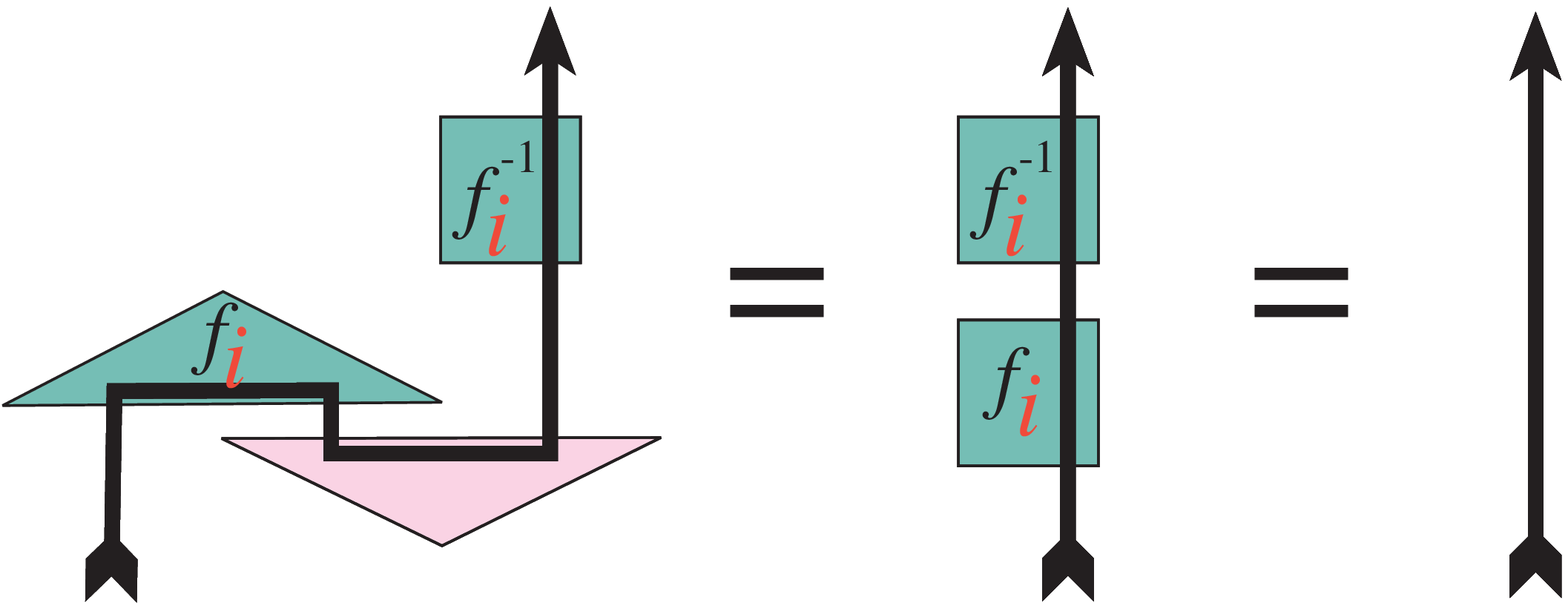,width=220pt}}
\end{minipage}

\smallskip\noindent
where the \em classical communication \em is now implicit in the fact that the index $i$ is both present in the costate (= measurement-branch) and the correction, and hence needs to be sent from Alice to Bob.

The classical communication can be made \emph{explicit} as a fully fledged part of the formalism, using \emph{additional types}: biproducts in \cite{AbrCoe2}, and  ``classical objects'' in \cite{CoePav}. This allows entire protocols, including the interplay between quantum and classical information which is often their most subtle ingredient, to be captured and reasoned about rigorously in a single formal framework.

\subsection{Remarks}
We close this Section with some remarks. We have seen that the categorical and diagrammatic setting for Quantum Mechanics developed by Abramsky and Coecke is strikingly close to that in which the Temperley-Lieb category lives. The main \emph{difference} is the free recourse to symmetry allowed in the Abramsky-Coecke setting (and in the main intended models for that setting, namely finite-dimensional Hilbert spaces with linear or completely positive maps). However, it is interesting to note that in the various protocols and constructions in Quantum Information and Computation which have been modelled in that setting to date \cite{AbrCoe2}, the symmetry has not played an essential role. The example of teleportation given above serves as an example.

This raises some natural questions:
\begin{center}
\textbf{How much of QM/QIC lives in the plane?}
\end{center}

More precisely:
\begin{itemize}
\item Which protocols make essential use of symmetry?
\item How much computational or information-processing power does the non-symmetric calculus have?
\item Does \emph{braiding} have some computational significance? (Remembering that between \emph{pivotal} and \emph{symmetric} we have \emph{braided} strongly compact closed categories) \cite{FY}.
\end{itemize}

\section{Planar Geometry of Interaction and the Temperley-Lieb Algebra}

We now address the issue of giving what, so far as I know, is the first \emph{direct}---or ``fully abstract''---description of the Temperley-Lieb category. Since the category $\TL$ is directly and simply described as the free $R$-linear category generated by $\DD$, we focus on the direct description of $\DD$.

\vsa
\textbf{Previous descriptions:}
\begin{itemize}
\item Algebraic, by generators and relations - whether ``locally'', of the Temperley-Lieb algebras $\TLA{\tau}{n}$, as in Jones' presentation, or ``globally'', by a description of $\DD$ as the free pivotal category, as in Proposition \ref{freepivot}.
\item Kauffman's topological description: diagrams ``up to planar isotopy''.
\end{itemize}

In fact, it is well known (see e.g. \cite{KauDiag}) that the diagrams are completely characterized by how the dots are joined up --- \ie by discrete relations on finite sets. This leaves us with the problem of how to capture
\begin{enumerate}
\item Planarity
\item The multiplication of diagrams --- \ie composition in $\DD$
\end{enumerate}
purely in terms of the data given by these relations.

The answers to these questions exhibit the connections that exist between the Temperley-Lieb category and what is commonly known as the \emph{``Geometry of Interaction''}. This is a dynamical/geometrical interpretation of proofs and Cut Elimination initiated by Girard \cite{GoI} as an off-shoot of Linear Logic \cite{LL}. The general setting for these notions is now known to be that of \emph{traced monoidal} and \emph{compact closed} categories --- in particular, in the free construction of compact closed categories over traced monoidal categories \cite{Retrac,AHS}. In fact, this general construction was first clearly described in \cite{JSV}, where one of the leading motivations was the knot-theoretic context.

Our results in this Section establish a two-way connection. In one direction, we shall use ideas from Geometry of Interaction to answer Question~2 above: that is, to define path composition (including the formation of loops) purely in terms of the discrete relations tabulating how the dots are joined up.
In the other direction, our answer to Question~1 will allow us to consider a natural \emph{planar variant} of the Geometry of Interaction.

\subsection{Some preliminary notions}

\subsubsection{Partial Orders}
We use the notation $P = (|P|, \leq_{P})$ for partial orders. Thus $|P|$ is the underlying set, and $\leq_{P}$ is the order relation (reflexive, transitive and antisymmetric) on this set. An order relation is \emph{linear} if for all $x, y \in |P|$, $x \leq_{P} y$ or $y \leq_{P} x$.

Given a natural number $n$, we define
$[n] := \{ 1 < \cdots < n \}$,
the linear order of length $n$. We define several constructions on partial orders. Given partial orders $P$, $Q$, we define:
\begin{itemize}
\item  The disjoint sum $P \oplus Q$, where $|P \oplus Q| = |P| + |Q|$,  the disjoint union of $|P|$ and $|Q|$, and 
\[ x \leq_{P\oplus Q} y \;\; \Longleftrightarrow \;\; (x \leq_P y) \; \vee \; (x \leq_Q y) . \]
\item The concatenation $P \lhd Q$,  where $|P \lhd Q| = |P| + |Q|$, with the following order:
\[ x \leq_{P\lhd Q} y \;\; \Longleftrightarrow \;\; (x \leq_P y) \; \vee \; (x \leq_Q y) \; \vee \; (x \in P \; \wedge \; y \in Q). \]
\item $P^{\opp} = (|P|, \geq_{P})$.
\end{itemize}
Given elements $x$, $y$ of a partial order $P$, we define:
\[ \begin{array}{ccc}
x \consis y & \;\; \Leftrightarrow \;\; &  (x \leq_P y) \; \vee \; (y \leq_P x) \\
x \incomp y & \;\; \Leftrightarrow \;\; & \neg (x \consis y) .
\end{array}
\]

\subsubsection{Relations}
A relation on a set $X$ is a subset of the cartesian product: $R \subseteq X \times X$. Since relations are sets, they are closed under unions and intersections. We shall also use the following operations of relation algebra:
\[ \begin{array}{lccl}
\mbox{\textbf{Identity relation:}} & \Idrel{X} & := & \{ (x,x) \mid x \in X \} \\
\mbox{\textbf{Relation composition:}} & R;S & := & \{ (x, z) \mid \exists y. \, (x, y) \in R \; \wedge \; (y, z) \in S \} \\
\mbox{\textbf{Relational converse:}} & R^c & := & \{ (y,x) \mid (x, y) \in R \} \\
\mbox{\textbf{Transitive closure:}} & R^+ & :=  & \bigcup_{k \geq 1} R^k \\
\mbox{\textbf{Reflexive transitive closure:}} & R^{\ast} & :=  & \bigcup_{k \geq 0} R^k
\end{array}
\]
Here $R^k$ is defined inductively: $R^0 := \Idrel{X}$, $R^1 := R$, $R^{k+1} := R;R^{k}$.
A relation $R$ is \emph{single-valued} or a \emph{partial function} if $R^c ; R \subseteq \Idrel{X}$. It is \emph{total} if $R;R^c \supseteq \Idrel{X}$. A \emph{function} $f : X \rightarrow X$ is a single-valued, total relation.

These notions extend naturally to relations $R \subseteq X \times Y$.

\subsubsection{Involutions}
A \emph{fixed-point free involution} on a set $X$ is a function $f : X \rightarrow X$ such that
\[ f^2 = 1_X , \qquad \qquad f \; \cap \; 1_X = \varnothing . \]
Thus for such a function $f(x) = y \; \Leftrightarrow \; x = f(y)$ and $f(x) \neq x$. We write $\Inv{X}$ for the set of fixed-point free involutions on a set $X$. Note that $\Inv{X}$ is \emph{not} closed under function composition; nor does it contain the identity function. We must look elsewhere for suitable notions of composition and identity.

An involution is equivalently described as a parition of $X$ into 2-element subsets:
\beq
X = \bigcup E , \qquad \qquad \mbox{where} \; E = \{ \{ x, y \} \mid f(x) = y \} . 
\eeq
This defines an undirected graph $G_f = (X, E)$. Clearly $G_f$ is 1-regular \cite{Die}: each vertex has exactly one incident edge. Conversely, every  graph $G = (X, E)$ with this property determines a unique $f \in \Inv{X}$ with $G_f = G$.
Note that a finite set can only carry such a structure of its cardinality is even.

\subsection{Formalizing diagrams}
From our previous discussion, it is fairly clear how we will proceed to formalize morphisms $\nn \rarr \mm$ in $\DD$.
Given $n, m \in \Nat$, we define $\Node{n}{m} = [n] \oplus [m]$.
We visualize this partial order as

\vsa
\begin{center}
\input{./pictures/Nnmp.tex}
\end{center}

\vsa
We use the notation $i'$ to distinguish the elements of $[m]$ in this disjoint union from those of $[n]$, which are unprimed.
Note that the order on $\Node{n}{m}$ has an immediate spatial interpretation in the diagrammatic representation: $i < j$ just in case $i$ lies to the left of $j$ on either the top or bottom line of dots, corresponding to $[n]$ and $[m]$ respectively.

A diagram connecting up dots pairwise will be formalized as a map $f \in \Inv{|\Node{n}{m}|}$.
Such a map can be visualized by drawing \emph{undirected arcs} between the pairs of nodes $i$, $j$ such that $f(i) = j$.
\subsubsection{Example}
The map $f \in \Inv{|\Node{4}{2}|}$ such that
\[ f : 1 \leftrightarrow 2', \qquad 2 \leftrightarrow 4, \qquad 3 \leftrightarrow 1' \]
is depicted thus:
\begin{center}
\input{./pictures/TLexp.tex}
\end{center}

Our task is now is \emph{characterize those involutions which are planar}. The key idea is that this can be done using just the order relations we have introduced.

\subsection{Characterizing Planarity}
A map $f \in \Inv{|\Node{n}{m}|}$ will be called \emph{planar} if it satisfies the following two conditions, for all $i, j \in \Node{n}{m}$:
\[
\begin{array}{lccc}
(\mbox{PL1}) & i < j < f(i) & \; \Longrightarrow \; & i < f(j) < f(i) \\
(\mbox{PL2}) & f(i) \incomp i < j \incomp f(j) & \; \Longrightarrow \; & f(i) < f(j) .
\end{array}
\]
It is instructive to see which possibilities are \emph{excluded} by these conditions.

\subsubsection{First condition}
\[
\begin{array}{lccc}
(\mbox{PL1}) & i < j < f(i) & \; \Longrightarrow \; & i < f(j) < f(i) 
\end{array}
\]
(PL1) rules out
\begin{center}
\input{./pictures/PL1exap.tex}
\end{center}

\vsa
\noindent where $f(j) \incomp f(i)$, and also

\vsa
\begin{center}
\input{./pictures/PL1exbp.tex}
\end{center}
where $f(i) < f(j)$.

\subsubsection{Second condition}

\vsn
\[
\begin{array}{lccc}
(\mbox{PL2}) & f(i) \incomp i < j \incomp f(j) & \; \Longrightarrow \; & f(i) < f(j) .
\end{array}
\]
Similarly, (PL2) rules out
\begin{center}
\input{./pictures/PL2exp.tex}
\end{center}

\vsb
\noindent We write $\PP (n, m)$ for the set of planar maps in $\Inv{|\Node{n}{m}|}$.

\begin{proposition}
\label{planarprop}
\begin{enumerate}
\item  Every planar diagram satisfies the two conditions.
\item Every involution satisfying the two conditions can be drawn as a planar diagram.
\end{enumerate}

\end{proposition}

Rather than proving this directly, it is simpler, and also instructive, to reduce it to a special case.
We consider arrows in $\DD$ of the special form $I \rarr \nn$. Such arrows consist only of caps. They correspond to \emph{points}, or \emph{states} in the terminology of Section~5. 

Since the top row of dots is empty, in this case we have a linear order, and the premise of condition (PL2) can never arise. Hence planarity for such arrows is just the simple condition (PL1) --- which can be seen to be equivalent to saying that, if we write a left parenthesis for each left end of a cap, and a right parenthesis for each right end, we get a well-formed string of parentheses. Thus
\begin{center}
\input{./pictures/paren.tex}
\end{center}
corresponds to
\[ () (()) . \]
(Of course, exactly similar comments apply to arrows of the form $\nn \rarr I$, \ie \emph{costates}.)
It is also clear\footnote{With an implicit appeal to the Jordan Curve Theorem!} that Proposition~\ref{planarprop} holds for such arrows.

Now we recall that quite generally, in any pivotal category we have the Hom-Tensor adjunction
\[
A\otimes B^*\!\rTo^{\coname{f}} I
\ \ \ \stackrel{\simeq}{\longleftrightarrow}\ \ \
A\rTo^{f}B
\ \ \ \stackrel{\simeq}{\longleftrightarrow}\ \ \
I \rTo^{\name{f}} A^*\!\otimes B
\]
\[
\uu f\uuu= (1_{A^*}\otimes f)\circ\eta_A: I \to A^*\otimes B
\quad \dd f\ddd=\epsilon_{B}\circ(f\otimes 1_{B^*}):A\otimes B^*\to I .
\]
We call $\uu f\uuu$ the \emph{name} of $f$, and $ \dd f\ddd$ the \emph{coname}.
The inverse to the map $f \mapsto \uu f \uuu$ is defined by
\[ g : I \rarr A^{*} \otimes B \;\; \mapsto \;\; (\epsilon_{A} \otimes 1_{B}) \circ (1_{A} \otimes g) : A \rarr B . \]
For example, we compute the name of the left wave:
\begin{center}
\input{./pictures/name.tex}
\end{center}
Applying the inverse transformation:
\begin{center}
\input{./pictures/unname.tex}
\end{center}
Note also that \emph{the unit is the name of the identity}: $\eta_{\nn} = \uu 1_{\nn} \uuu$, and similarly $\epsilon_{\nn} = \dd 1_{\nn} \ddd$.

Thus we see that diagrammatically, the process of forming the name of an arrow involves reversing the left-right order of the top row of dots by rotating them concentrically, and sliding them down to lie parallel with, and to the left of, the bottom row. In this process, cups are turned into caps, while through lines are stretched out and turned to also form caps.

This transposition of the top row of dots can be described order-theoretically, as replacing the partial order $[n] \oplus [m]$ by the linear order $[n]^{\opp} \lhd [m]$. Note that the underlying sets of these two partial orders are the same: $|[n] \oplus [m]| = |[n]^{\opp} \lhd [m]|$. Thus $\uu f \uuu$ is essentially \emph{the same function} as $f$.

\begin{proposition}
For $f \in \Inv{|\Node{n}{m}|}$, the following are equivalent:
\begin{enumerate}
\item $f$ satisfies (PL1) and (PL2) with respect to $[n] \oplus [m]$.
\item $f$ satisfies (PL1) with respect to $[n]^{\opp} \lhd [m]$. 
\end{enumerate}
\end{proposition}
\begin{proof}
Firstly, assume (2), and suppose $f(i) \incomp i < j \incomp f(j)$. If $i < j$ in the bottom  row, then $f(j) < i < j$ in $[n]^{\opp} \lhd [m]$, so by (PL1), $f(j) < f(i) < j$, \ie $f(i) < f(j)$ in $[n] \oplus [m]$, as required.
Now suppose $i < j$ in the top row. Then $j < i < f(j)$ in $[n]^{\opp} \lhd [m]$, so by (PL1), $j < f(i) < f(j)$, and in particular $f(i) < f(j)$.

Now assume (1), and suppose that $i < j < f(i)$ in $[n]^{\opp} \lhd [m]$.
The interesting case is where $i$ is in the top row and $f(i)$ in the bottom row.
We need to do some case analysis. Suppose firstly that $j$ is in the top row.
If $f(j)$ is in the bottom row, then $f(j) \incomp j < i \incomp f(i)$ in $[n] \oplus [m]$, and we can apply (PL2) to conclude that $f(j) < f(i)$, and hence $i < f(j) < f(i)$ in $[n]^{\opp} \lhd [m]$.
If $f(j)$ is in the top row, we must have $f(j) < i$ by (PL1) for  $[n] \oplus [m]$, and hence $i < f(j) < f(i)$ in $[n]^{\opp} \lhd [m]$.

Now suppose that $j$ is in the bottom row. If $f(j)$ is in the bottom row, we must have $f(j) < f(i)$ by (PL1). If $f(j)$ is in the top row, then we have $f(j)  \incomp j < f(i) \incomp i$
in $[n] \oplus [m]$, so by (PL2) we have $f(j) < i$, and hence
$i < f(j) < f(i)$  in $[n]^{\opp} \lhd [m]$.
\end{proof}
Since (PL1) characterizes planarity for $\uu f \uuu$, it  follows that (PL1) and (PL2) characterize planarity for $f$.

\subsection{The Temperley-Lieb Category}

Our aim is now to define a category $\TLC$, which will yield the desired description of the diagrammatic category $\DD$.
The \emph{objects} of $\TLC$ are the natural numbers.
The homset $\TLC (\nn,\mm)$ is defined to be the cartesian product $\Nat \times \PP (n, m)$. Thus a morphism $\nn \rightarrow \mm$ in $\TLC$ consists of a pair $(k, f)$, where $k$ is a natural number, and $f \in \PP (n, m)$ is a planar map in $\Inv{|\Node{n}{m}|}$. The idea is that $k$ is a counter for the number of loops, so such an arrow can be written $\delta^{k}\cdot f$ in the notation used previously.

It remains to define the composition and identities in this category. Clearly (even leaving aside the natural number components of morphisms)  composition cannot be defined as ordinary function composition.  This does not even make sense --- the codomain of an involution $f \in \PP(n, m)$ does not match the domain of an involution $g \in \PP(m, p)$ --- let alone yield a function with the  necessary properties to be a morphism in the category.

\subsubsection{Composition: The ``Execution Formula''}
Consider a map $f : [n] + [m] \longrightarrow [n] + [m]$. Each input lies in \emph{either} $[n]$ \emph{or} $[m]$ (exclusive or), and similarly for the corresponding output. This leads to a decomposition of $f$ into four \emph{disjoint partial maps}:
\[ \begin{array}{ll}
f_{n,n} : [n] \longrightarrow [n]  & \qquad f_{n,m} : [n] \longrightarrow [m] \\
f_{m,n} : [m] \longrightarrow [n]  & \qquad f_{m,m} : [m] \longrightarrow [m] 
\end{array} \]
so that $f$ can be recovered as the disjoint union of these four maps. 
If $f$ is an involution, then these maps will be partial involutions.

Note that these components have a natural diagrammatic reading: $f_{n,n}$ describes the \emph{cups} of $f$, $f_{m,m}$ the \emph{caps}, and $f_{n,m} = f_{m,n}^{c}$ the \emph{through lines}.

Now suppose we have maps $f : [n] + [m] \rightarrow [n] + [m]$ and $g : [m] + [p] \rightarrow [m]+[p]$.
We write the decompositions of $f$ and $g$ as above in matrix form:
\[ f = \left( \begin{array}{cc}
 f_{n,n} & f_{n, m} \\
f_{m, n} & f_{m, m}
\end{array} \right)
\qquad \qquad
g =   \left( \begin{array}{cc}
g_{m,m} & g_{m, p} \\
g_{p, m} & g_{p, p}
\end{array} \right)
\]

\noindent We can view these maps as \emph{binary relations} on $[n] + [m]$ and $[m] + [p]$ respectively, and use relational algebra (union $R \cup S$, relational composition $R ; S$ and reflexive transitive closure $R^{\ast}$) to define a \emph{new relation} $\theta$ on $[n] + [p]$. If we write
\[ \theta =     \left( \begin{array}{cc} 
\theta_{n,n} &  \theta_{n, p} \\
\theta_{p, n} & \theta_{p, p}
\end{array} \right)
\]
so that $\theta$ is the disjoint union of these four components, then we can define it component-wise as follows:
\[ \begin{array}{lcl}
\theta_{n,n} & \;\; = \;\; & 
f_{n,n} \; \cup \;\; f_{n,m} ; g_{m,m} ; (f_{m,m} ; g_{m,m})^{\ast} ; f_{m,n} \\
\theta_{n,p} & \;\; = \;\; & f_{n,m}; (g_{m,m} ; f_{m,m})^{\ast} ; g_{m,p} \\
\theta_{p,n} & \;\; = \;\; & g_{p,m}; (f_{m,m} ; g_{m,m})^{\ast} ; f_{m,n} \\
\theta_{p,p} & \;\; = \;\; & g_{p,p} \; \cup \;\; g_{p,m} ; f_{m,m} ; (g_{m,m} ; f_{m,m})^{\ast} ; g_{m,p} .
\end{array}
\]
We can give clear intuitive readings for how these formulas express composition of  paths in diagrams in terms of relational algebra:
\begin{itemize}
\item The component $\theta_{n,n}$ describes the \emph{cups} of the diagram resulting from the composition. These are the union of the cups of $f$ ($f_{n,n}$), together with paths that start from the top row with a through line of $f$, given by $f_{n,m}$, then go through an alternating odd-length sequence of cups of $g$ ($g_{m,m}$) and caps of $f$ ($f_{m,m}$), and finally return to the top row by a through line of $f$ ($f_{m,n}$).
\begin{center}
\input{./pictures/cupcomp.tex}
\end{center}
\item Similarly, $\theta_{p,p}$ describes the caps of the composition.
\item $\theta_{n,p} = \theta_{p,n}^{c}$ describe the through lines. Thus $\theta_{n,p}$ describes paths which start with a through line of $f$ from $n$ to $m$, continue with an alternating even-length (and possibly empty) sequence of cups of $g$ and caps of $f$, and finish with a through line of $g$ from $m$ to $p$. 
\begin{center}
\input{./pictures/throughcomp.tex}
\end{center}
All through lines from $n$ to $p$ must have this form.
\end{itemize}
This formula corresponds to the interpretation of Cut-Elimination  in the Geometry of Interaction interpretation of proofs in Linear Logic (and by extension in related logics and type theories) \cite{GoI}. A more abstract and general perspective on how this construction arises can be given in the setting of traced monoidal categories \cite{JSV,Retrac}.

\begin{proposition}
If $f$ and $g$ are planar, so is $\theta$.
\end{proposition}
We write $\theta = \comp{f}{g} \in \PP (n, p)$

\subsubsection{Cycles}
Given $f \in \PP (n, m)$,  $g \in \PP (m, p)$, we define $\Cyc{f}{g} := f_{m,m} ; g_{m,m}$. Note that $\Cyc{f}{g}^c = (g_{m,m} ; f_{m,m})$, and 
\[ \Cyc{f}{g} ; \Cyc{f}{g}^c \subseteq \Idrel{[m]} , \qquad\qquad
\Cyc{f}{g}^c ; \Cyc{f}{g} \subseteq \Idrel{[m]} . \] 
Thus $\Cyc{f}{g}$ is a \emph{partial bijection}. However, in general it is neither an involution, nor  fixpoint-free.
The \emph{cyclic elements} of $\Cyc{f}{g}$ are those elements of $[m]$ which lie in the intersection 
\[ \Cyc{f}{g}^{+} \; \cap \; \Idrel{[m]}.
\]
\begin{center}
\input{./pictures/cyclic.tex}
\end{center}
Thus if $i$ is a cyclic element, there is a least $k > 0$ such that
$\Cyc{f}{g}^{k}(i) = i$. The corresponding \emph{cycle} is
\[ \{ i, \, \Cyc{f}{g}(i), \, \ldots , \, \Cyc{f}{g}^{k-1}(i) \} . \]
Distinct cycles are disjoint.
We write $Z(f, g)$ for the number of distinct cycles of $\Cyc{f}{g}$.

\subsubsection{Composition and Identities}
Finally, we define the composition of morphisms in $\TLC$. Given $(s, f) : \nn \rightarrow \mm$ and $(t, g) : \mm \rightarrow \pp$:
\[ (t, g) \circ (s, f) = (s+t+Z(f,g), \comp{f}{g}). \]

The identity morphism $\id_{\nn} : \nn \rightarrow \nn$ is defined to be the pair $(0, \tau_{n,n})$, where $\tau_{n,n}$ is the \emph{twist map} on $[n] + [n]$; i.e. the involution
$i \leftrightarrow i'$. Diagrammatically, this is just
\begin{center}
\input{./pictures/TLid.tex}
\end{center}

\vsa
\noindent Note that this is \emph{not} the identity map on $[n] + [n]$ --- indeed it is (necessarily) fixpoint free!

\begin{proposition}
$\TLC$ with composition and identities as defined above is a category.
\end{proposition}

\subsubsection{$\TLC$ as a pivotal category}
The monoidal structure of $\TLC$ is straightforward. If $(k, f) : \nn \rarr \mm$ and $(l, g) : \pp \rarr \qq$, then $(k+l, f+g) : \nn+\pp \rarr \mm+\qq$, where $f+g \in \PP(n+p, m+q)$ is the evident disjoint union of the involutions $f$ and $g$.

The unit $\eta_{\nn} : I \rarr \nn+\nn$ is given by
\[ i \leftrightarrow i' \quad (1 \leq i \leq n) , \]
and similarly for the counit. Note that identities, units and counits are all essentially the same maps, but with distinct \emph{types}, which partition their arguments between inputs and outputs differently.

We describe the dual, adjoint and conjugate of an arrow $(k, f) : \nn \rarr \mm$.
Let $\tau_{n,m} : [n]+[m] \isoarrow [m] + [n]$ be the symmetry isomorphism of the disjoint union, and 
\[ \rho_{n} :  [n] \isoarrow [n] :: i \mapsto n-i+1 \]
be the order-reversal isomorphism. Note that
\[ \tau_{n,m}^{-1} = \tau_{m, n}, \quad \rho_{n}^{-1} = \rho_{n}, \quad \tau_{n,m} \circ (\rho_{n}+\rho_{m}) = (\rho_{m}+\rho_{n}) \circ \tau_{n, m} . \]
Then we have $(k, f)^{\bullet} = (k, f^{\bullet})$, where:
\[ \begin{array}{rcl}
f^{\dagger} & = & \tau_{n.m} \circ f \circ \tau_{n,m}^{-1} \\
f_{*} & = & (\rho_{n}+\rho_{m}) \circ f \circ (\rho_{n}+\rho_{m})^{-1} \\
f^{*} & = & (f^{\dagger})_{*} .
\end{array}
\]

\subsubsection{The Main Result}
\begin{theorem}
$\TLC$ is isomorphic as a strict, pivotal dagger category to $\DD$.
\end{theorem}

\noindent As an immediate Corollary of this result and Proposition~\ref{freepivotdagprop}, we have:
\begin{theorem}
$\TLC$ is the free strict, pivotal dagger category on one self-dual generator.
\end{theorem}
This is in the same spirit as the characterizations of free compact and dagger compact categories in \cite{KL,Abr05}.

These results can easily be extended to descriptions of the free pivotal dagger category over an arbitrary generating category, leading to \emph{oriented Temperley-Lieb algebras} with \emph{primitive (physical) operations}. We refer to \cite{Abr05} for a more detailed presentation (in the symmetric case).

\section{Planar $\lambda$-Calculus}

Our aim in this section is to show how a restricted form of $\lambda$-calculus can be interpreted in the Temperley-Lieb category, and how $\beta$-reduction of $\lambda$-terms, which is an important foundational paradigm for computation, is then reflected diagrammatically as geometric simplification, \ie ``yanking lines straight''. We can give only a brief indication of what is in fact a rich topic in its own right.
See \cite{AD04,AHS,KnoLog} for discussions of related matters.

\subsection{The $\lambda$-Calculus}
We begin with a (very) brief review of the $\lambda$-calculus \cite{Church,Bar84}, which is an important foundational paradigm in Logic and Computation, and in particular forms the basis for all modern functional programming languages.

The syntax of the $\lambda$-calculus is beguilingly simple. Given a set of variables $x$, $y$, $z$, \ldots we define the set of terms as follows:
\[ t \;\; ::= \;\; x \; \mid \; \underbrace{tu}_{\mbox{application}}
\; \mid \;
  \underbrace{\lambda x . \, t}_{\mbox{abstraction}} \]
 \noindent \textbf{Notational Convention:} We write
 \[ t_{1}t_{2}\cdots t_{k} \equiv (\cdots (t_{1}t_{2})\cdots )t_{k} . \]
Some examples of terms:
\begin{center}
\begin{tabular}{ll}
$\lambda x . \, x$ & identity function \\
$ \lambda f . \, \lambda x. \, fx$ & application \\
$\lambda f. \, \lambda x. \, f(fx)$ & double application \\
$\lambda f. \, \lambda g. \, \lambda x. \, g(f(x))$ & composition $g
\circ f$ \\
\end{tabular}
\end{center}
The basic equation governing this calculus is
\emph{$\beta$-conversion}:
\[ (\lambda x. \, t)u = t[u/x] \]
E.g. (assuming some arithmetic operations are given)
\[ (\lambda f. \, \lambda x. \, f(f x))(\lambda x. \, x+1) 0 \; = \; (0+1)+1 \; = \;
2. \]
By orienting this equation, we get a `dynamics' ---
\emph{$\beta$-reduction}
\[ (\lambda x. \, t)u \rightarrow t[u/x] \]
Despite its sparse syntax, $\lambda$-calculus is very expressive---it is in fact a universal model of computation,  equivalent to Turing machines.

\subsection{Types}
One important way of constraining the $\lambda$-calculus is to  introduce Types.
\begin{center}
\fbox{{\textbf{Types are there to stop you doing (bad) things}}}
\end{center}
Types are in fact one of the most fruitful \emph{positive} ideas in Computer Science!

We shall introduce a (highly restrictive) type system, such that the typable terms can be interpreted in the Temperley-Lieb category (in fact, in any pivotal category).

Firstly, assuming some set of \emph{basic types} $B$, we define a syntax of general types:
\[ T \;\; ::= \;\; B \; \mid \; T \rightarrow T . \]
Intuitively, $T \rarr U$ represents the type of functions which take inputs of type $T$ to outputs of type $U$.

\noindent \textbf{Notational Convention:} We write
\[ T_{1} \rarr T_{2} \rarr \cdots T_{k} \rarr T_{k+1} \quad \equiv \quad T_{1} \rarr (T_{2} \rarr  \cdots (T_{k} \rarr T_{k+1})\cdots ) . \]

\noindent \textbf{Examples:}
\[ A \rightarrow A \rightarrow A \qquad \mbox{first-order 
  function type} \]
\[ (A \rightarrow A ) \rightarrow A \qquad
\mbox{second-order function type} \]

We now introduce a formal system  for deriving \emph{typing judgements}, of the form:
\[ x_1 : T_1 , \ldots x_k : T_k \vdash t : T . \]
Such a judgement asserts that the term $t$ has type $T$ \emph{under the assumption} (or: \emph{in
  the context}) that the variable $x_1$ has type $T_1$, \ldots ,
$x_k$ has type $T_k$. All the variables $x_{i}$ appearing in the context must be distinct --- and in our setting, the order in which the variables appear in the list is significant.

\noindent There is one basic form of axiom, for typing variables:

\textbf{Variable}
\[ \infer{x:T \vdash  x:T}{} \]
and two inference rules, for typing abstractions and applications respectively:

\textbf{Function}
\[ \infer{\Gamma \vdash \lambda x. \, t: U \rightarrow T}{\Gamma , x:U \vdash t:T}
\qquad
\infer{\Gamma, \Delta \vdash tu : T}{\Gamma \vdash t : U \rightarrow T \qquad \Delta
  \vdash u:U}
\]
Note that $\Gamma, \Delta$ represents the concatenation of the lists $\Gamma$, $\Delta$. This implies that the variables appearing in $\Gamma$ and $\Delta$ are distinct---an important \emph{linearity constraint} in the sense of Linear Logic \cite{LL}.

\subsection{Interpretation in Pivotal Categories}
We now show how terms typable in our system can be interpreted in a pivotal category $\CC$. We assume firstly that the basic types $B$ have been interpreted as objects $\lsem B \rsem$ of $\CC$. We then extend this to general types by:
\[ \lsem T \rarr U \rsem = \lsem U \rsem \otimes \lsem T \rsem^{*} . \]

Now we show how, for each typing judgement $\Gamma \vdash t : T$, to assign an arrow
\[ \lsem \Gamma \rsem \lrarr \lsem T \rsem , \]
where if $\Gamma = x_1 : T_1 , \ldots x_k : T_k$, 
\[ \lsem \Gamma \rsem = \lsem T_{1} \rsem \otimes \cdots \otimes \lsem T_{k} \rsem . \]
This assignment is defined by induction on the derivation of the typing judgement in the formal system.

\noindent \textbf{Variable}
\[ \infer{x : T \vdash  x:T}{} 
\qquad \qquad \qquad
\infer{1_{\lsem T \rsem} : \lsem T \rsem \longrightarrow \lsem T \rsem}{}
\]

\noindent \textbf{Abstraction}

To interpret $\lambda$-abstraction, we use the adjunction
\[ \Curryr : \CC(A \otimes B, C) \simeq \CC(A, C \otimes B^{*}) \]
\[ \Curryr(f) = \begin{diagram}
A & \rTo^{1_{A} \otimes \eta_{B*}} & A \otimes B \otimes B^{*} & \rTo^{f \otimes 1_{B^{*}}} & C \otimes B^{*}
\end{diagram}
\]
We can then define:
\[ \infer{\Gamma \vdash \lambda x:U. \, t: U \rightarrow T}{\Gamma , x:U \vdash t:T}
\qquad \qquad
\infer{\Curryr(\lsem t \rsem) : \lsem \Gamma  \rsem \longrightarrow \lsem T \rsem \otimes \lsem U \rsem^{*}}{\lsem t \rsem :
 \lsem   \Gamma \rsem \otimes \lsem U \rsem \longrightarrow \lsem T \rsem}
\]

\noindent \textbf{Application}

We use the following operation of \emph{right application}:
\[ \Rapp : \CC(C, B \otimes A^{*}) \times \CC(D, A) \lrarr \CC(C \otimes D, B) \]
\[ \Rapp(f, g) = \begin{diagram}
C \otimes D & \rTo^{f \otimes g} &  B \otimes A^{*} \otimes A & \rTo^{1_{B} \otimes \epsilon_{B^{*}}} & B .
\end{diagram}
\]
We can then define:
\[ \infer{\Gamma, \Delta \vdash tu : T}{\Gamma \vdash t : U \rightarrow T \qquad \Delta
  \vdash u:U}
\qquad 
\infer{\Rapp( \lsem t \rsem, \lsem u \rsem ) : \lsem  \Gamma \rsem \otimes \lsem \Delta \rsem \longrightarrow 
  \lsem T \rsem}{\lsem t \rsem :  \lsem \Gamma \rsem \longrightarrow \lsem T \rsem \otimes \lsem U \rsem^{*} \qquad \lsem u \rsem :  \lsem \Delta \rsem 
  \longrightarrow \lsem U \rsem}
\]

\noindent It can be proved that this interpretation \emph{is sound for $\beta$-conversion}, \ie
\[ \lsem (\lambda x. \, t)u \rsem = \lsem t[u/x] \rsem \]
in any pivotal category.

\subsection{An Example}
We now discuss an example to show how all this works diagrammatically in $\TLC$.
We shall consider the \emph{bracketing combinator}
\[ \Bcomb \equiv \lambda x. \lambda y. \lambda z. \, x(yz) . \]
This is characterized by the equation
\[ \Bcomb a b c = a (bc) . \]
Firstly, we derive a typing judgement for this term:
\begin{prooftree}
\def\fCenter{\mbox{\ $\vdash$\ }}
\AxiomC{$x:B \rarr C \vdash x:B \rarr C$}
\AxiomC{$y:A \rarr B \vdash y:A \rarr B$}
\AxiomC{$z:A  \vdash z:A$}
\BinaryInfC{$y:A \rarr B, z:A \fCenter yz : B$}
\BinaryInfC{$x:B \rarr C, y:A \rarr B, z:A \fCenter x(yz) : C$}
\UnaryInfC{$x:B \rarr C, y:A \rarr B  \fCenter \lambda z. \, x(yz) : A \rarr C$}
\UnaryInfC{$x:B \rarr C \fCenter \lambda y. \lambda z. \, x(yz) : (A \rarr B) \rarr (A \rarr C)$}
\UnaryInfC{$\fCenter \lambda x. \lambda y. \lambda z. \, x(yz) : (B \rarr C) \rarr (A \rarr B) \rarr (A \rarr C)$}
\end{prooftree}
Now we take $A = B = C = \mathbf{1}$ in $\TLC$. The interpretation of the open term 
\[  x:B \rarr C, y:A \rarr B, z:A \vdash x(yz) : C \]
is as follows:
\begin{center}
\input{./pictures/Bopen.tex}
\end{center}
Here $x^{+}$ is the output of $x$, and $x^{-}$ the input, and similarly for $y$. The output  of the whole expression is $o$.
When we abstract the variables, we obtain the following  caps-only diagram:
\begin{center}
\input{./pictures/Bclosed.tex}
\end{center}

\vsa
\noindent Now we consider an application $\Bcomb a b c$:
\begin{center}
\input{./pictures/Bapp.tex}
\end{center}

\subsection{Discussion}
The typed $\lambda$-calculus we have used here is in fact a fragment of the \emph{Lambek calculus} \cite{Lam58}, a basic non-commutative logic and $\lambda$-calculus, which has found extensive applications in computational linguistics \cite{Busz,Moort}. The Lambek calculus can be interpreted in any monoidal biclosed category, and has notions of \emph{left abstraction and application}, as well as the right-handed versions we have described here. Pivotal categories have stronger properties than monoidal biclosure; for example, duality and adjoints allow the left- and right-handed versions of abstraction and application to be defined in terms of each other. Moreover, the duality means that the corresponding logic has a \emph{classical} format, with an involutive negation.
Thus there is much more to this topic than we have had the time to discuss here. We merely hope to have given an impression of how the geometric ideas expressed in the Temperley-Lieb category have natural connections to a central topic in Logic and Computation.

\section{Further Directions}

We hope to have given an indication of the rich and suggestive connections which exist between ideas stemming from knot theory, topology and mathematical physics, on the one hand, and logic and computation on the other, with the Temperley-Lieb category serving as an intuitive and compelling meeting point. We hope that further investigation will uncover deeper links and interplays, leading to new insights in both directions.

We conclude with a few specific directions for future work:

\begin{itemize}
\item The symmetric case,  where we drop the planarity constraint, is also interesting. The algebraic object corresponding to the Temperley-Lieb algebra in this case is  the \emph{Brauer algebra} \cite{Brauer}, important in the representation theory of the Orthogonal group (Schur-Weyl duality).
Indeed, there are now a family of various kinds of diagram algebras: partition algebras, rook algebras etc., arising in quantum statistical mechanics, and studied in Representation Theory \cite{HalvRam}. 
\item The categorical perspective suggests \emph{oriented} versions of the Temperley-Lieb algebra and related structures, where we no longer have $A = A^{*}$. This is also natural from the point of view of Quantum Mechanics, where this non-trivial duality on objects distinguishes complex from  real Hilbert spaces.
\item We can ask how expressive planar Geometry of Interaction is; and what r\^ole may be played by braiding or other geometric information.
\item Again, it would be interesting to understand the scope and limits of planar Quantum Mechanics and Quantum information processing.
\end{itemize}

\bibliographystyle{plain}
\bibliography{tam}

\end{document}

%% file: macros.tex
\newcommand{\rarr}{\rightarrow}
\newcommand{\lrarr}{\longrightarrow}
\newcommand{\lsem}{\llbracket}
\newcommand{\rsem}{\rrbracket}
\newcommand{\Curryr}{\Lambda_{\mathsf{r}}}
\newcommand{\Rapp}{\mathsf{RApp}}
\newcommand{\Bcomb}{\mathbf{B}}
\newcommand{\nn}{\mathbf{n}}
\newcommand{\mm}{\mathbf{m}}
\newcommand{\pp}{\mathbf{p}}
\newcommand{\qq}{\mathbf{q}}
\newcommand{\kk}{\mathbf{k}}
\newcommand{\zz}{\mathbf{0}}
\newcommand{\beq}{\begin{equation}}
\newcommand{\eeq}{\end{equation}\par\noindent}
\newcommand{\PP}{\mathcal{P}}

\newcommand{\CC}{\mathcal{C}}
\newcommand{\DD}{\mathcal{D}}
\newcommand{\II}{\mathrm{I}}

\newcommand{\dd}{\llcorner}

\newcommand{\ddd}{\lrcorner}
\newcommand{\uu}{\ulcorner}
\newcommand{\uuu}{\urcorner}
\newcommand{\Nat}{\mathbb{N}}

\newcommand{\ie}{\textit{i.e.}~}

\newcommand{\isoarrow}{\stackrel{\cong}{\longrightarrow}}
\newcommand{\id}{\mathsf{id}}

\newcommand{\consis}{\uparrow}
\newcommand{\incomp}{\;\#\;}
\newcommand{\Inv}[1]{\mathsf{Inv}(#1)}
\newcommand{\Node}[2]{\mathsf{N}(#1,#2)}
\newcommand{\TL}{\mathcal{T}}
\newcommand{\TLM}[1]{\mathcal{M}_{#1}}
\newcommand{\TLA}[2]{\mathcal{A}_{#2}(#1)}
\newcommand{\TLC}{\mathbf{TL}}
\newcommand{\comp}[2]{#2 \odot #1}
\newcommand{\Idrel}[1]{1_{#1}}
\newcommand{\Cyc}[2]{\chi (#1,#2)}

\newcommand{\Trace}[1]{\mathsf{Tr}(#1)}

%% file: pictures/TLgens.tex
\psset{unit=1in,cornersize=absolute,dimen=middle}%
\begin{pspicture}(-0,-0.973848)(2.782422,0.026152)%
\newgray{fillval}{0.3}
\pscircle[fillstyle=solid,fillcolor=fillval](0.015458,0){0.015458}
\pscircle[fillstyle=solid,fillcolor=fillval](0.200953,0){0.015458}
\pscircle[fillstyle=solid,fillcolor=fillval](0.386447,0){0.015458}
\rput(0.571942,0){$\cdots$}
\pscircle[fillstyle=solid,fillcolor=fillval](0.757437,0){0.015458}
\pscircle[fillstyle=solid,fillcolor=fillval](0.015458,-0.649232){0.015458}
\pscircle[fillstyle=solid,fillcolor=fillval](0.200953,-0.649232){0.015458}
\pscircle[fillstyle=solid,fillcolor=fillval](0.386447,-0.649232){0.015458}
\rput(0.571942,-0.649232){$\cdots$}
\pscircle[fillstyle=solid,fillcolor=fillval](0.757437,-0.649232){0.015458}
\uput{0.5ex}[u](0.015458,0.015458){$1$}
\uput{0.5ex}[u](0.200953,0.015458){$2$}
\uput{0.5ex}[u](0.386447,0.015458){$3$}
\uput{0.5ex}[u](0.757437,0.015458){$n$}
\uput{0.5ex}[d](0.015458,-0.66469){$1$}
\uput{0.5ex}[d](0.200953,-0.66469){$2$}
\uput{0.5ex}[d](0.386447,-0.66469){$3$}
\uput{0.5ex}[d](0.757437,-0.66469){$n$}
\psarc(0.108205,-0.015458){0.092747}{-179.999999}{-0.000001}
\psarcn(0.108205,-0.633774){0.092747}{180}{0}
\psline(0.386447,-0.015458)(0.386447,-0.633774)
\psline(0.757437,-0.015458)(0.757437,-0.633774)
\rput(0.386447,-0.973848){$U_1$}
\rput(1.391211,-0.309158){$\cdots$}
\pscircle[fillstyle=solid,fillcolor=fillval](2.024985,0){0.015458}
\rput(2.210479,0){$\cdots$}
\pscircle[fillstyle=solid,fillcolor=fillval](2.395974,0){0.015458}
\pscircle[fillstyle=solid,fillcolor=fillval](2.581469,0){0.015458}
\pscircle[fillstyle=solid,fillcolor=fillval](2.766964,0){0.015458}
\pscircle[fillstyle=solid,fillcolor=fillval](2.024985,-0.649232){0.015458}
\rput(2.210479,-0.649232){$\cdots$}
\pscircle[fillstyle=solid,fillcolor=fillval](2.395974,-0.649232){0.015458}
\pscircle[fillstyle=solid,fillcolor=fillval](2.581469,-0.649232){0.015458}
\pscircle[fillstyle=solid,fillcolor=fillval](2.766964,-0.649232){0.015458}
\uput{0.5ex}[u](2.024985,0.015458){$1$}
\uput{0.5ex}[u](2.766964,0.015458){$n$}
\uput{0.5ex}[d](2.024985,-0.66469){$1$}
\uput{0.5ex}[d](2.766964,-0.66469){$n$}
\psarc(2.674216,-0.015458){0.092747}{180}{360}
\psarcn(2.674216,-0.633774){0.092747}{180}{0}
\psline(2.024985,-0.015458)(2.024985,-0.633774)
\psline(2.395974,-0.015458)(2.395974,-0.633774)
\rput(2.395974,-0.973848){$U_{n-1}$}
\end{pspicture}%

%% file: pictures/TLrel12.tex
\psset{unit=1in,cornersize=absolute}%
\begin{pspicture}(0,-0.989796)(3.306122,0.010204)
\psset{linewidth=1pt}%
\newgray{fillval}{0.3}
\pscircle[fillstyle=solid,fillcolor=fillval](0.010204,0){0.017123}
\pscircle[fillstyle=solid,fillcolor=fillval](0.132653,0){0.017123}
\pscircle[fillstyle=solid,fillcolor=fillval](0.255102,0){0.017123}
\pscircle[fillstyle=solid,fillcolor=fillval](0.010204,-0.22449){0.017123}
\pscircle[fillstyle=solid,fillcolor=fillval](0.132653,-0.22449){0.017123}
\pscircle[fillstyle=solid,fillcolor=fillval](0.255102,-0.22449){0.017123}
\pscircle[fillstyle=solid,fillcolor=fillval](0.010204,-0.44898){0.017123}
\pscircle[fillstyle=solid,fillcolor=fillval](0.132653,-0.44898){0.017123}
\pscircle[fillstyle=solid,fillcolor=fillval](0.255102,-0.44898){0.017123}
\pscircle[fillstyle=solid,fillcolor=fillval](0.010204,-0.673469){0.017123}
\pscircle[fillstyle=solid,fillcolor=fillval](0.132653,-0.673469){0.017123}
\pscircle[fillstyle=solid,fillcolor=fillval](0.255102,-0.673469){0.017123}
\psarc(0.071429,-0.010204){0.061224}{-179.999999}{-0.000001}
\psline(0.255102,-0.010204)(0.255102,-0.214286)
\psarcn(0.071429,-0.214286){0.061224}{180}{0}
\psline(0.010204,-0.234694)(0.010204,-0.438776)
\psarc(0.193878,-0.234694){0.061224}{180}{360}
\psarcn(0.193878,-0.438776){0.061224}{180}{0}
\psarc(0.071429,-0.459184){0.061224}{180}{360}
\psline(0.255102,-0.459184)(0.255102,-0.663265)
\psarcn(0.071429,-0.663265){0.061224}{180}{0}
\rput(0.571429,-0.346939){{\huge $=$}}
\pscircle[fillstyle=solid,fillcolor=fillval](0.887755,0){0.017123}
\pscircle[fillstyle=solid,fillcolor=fillval](1.010204,0){0.017123}
\pscircle[fillstyle=solid,fillcolor=fillval](1.132653,0){0.017123}
\pscircle[fillstyle=solid,fillcolor=fillval](0.887755,-0.673469){0.017123}
\pscircle[fillstyle=solid,fillcolor=fillval](1.010204,-0.673469){0.017123}
\pscircle[fillstyle=solid,fillcolor=fillval](1.132653,-0.673469){0.017123}
\psarc(0.94898,-0.010204){0.061224}{-179.999998}{-0.000002}
\psline(1.132653,-0.010204)(1.132653,-0.663265)
\psarcn(0.94898,-0.663265){0.061224}{179.999998}{0.000002}
\rput(0.561224,-0.989796){$U_{1}U_{2}U_{1} = U_{1}$}
\pscircle[fillstyle=solid,fillcolor=fillval](1.969388,0){0.017123}
\pscircle[fillstyle=solid,fillcolor=fillval](2.091837,0){0.017123}
\pscircle[fillstyle=solid,fillcolor=fillval](2.214286,0){0.017123}
\pscircle[fillstyle=solid,fillcolor=fillval](1.969388,-0.22449){0.017123}
\pscircle[fillstyle=solid,fillcolor=fillval](2.091837,-0.22449){0.017123}
\pscircle[fillstyle=solid,fillcolor=fillval](2.214286,-0.22449){0.017123}
\pscircle[fillstyle=solid,fillcolor=fillval](1.969388,-0.44898){0.017123}
\pscircle[fillstyle=solid,fillcolor=fillval](2.091837,-0.44898){0.017123}
\pscircle[fillstyle=solid,fillcolor=fillval](2.214286,-0.44898){0.017123}
\psarc(2.030612,-0.010204){0.061224}{180}{360}
\psline(2.214286,-0.010204)(2.214286,-0.214286)
\psarcn(2.030612,-0.214286){0.061224}{180}{0}
\psarc(2.030612,-0.234694){0.061224}{180}{360}
\psline(2.214286,-0.234694)(2.214286,-0.438776)
\psarcn(2.030612,-0.438776){0.061224}{180}{0}
\rput(2.632653,-0.22449){{\huge $=$}}
\pscircle[fillstyle=solid,fillcolor=fillval](3.05102,0){0.017123}
\pscircle[fillstyle=solid,fillcolor=fillval](3.173469,0){0.017123}
\pscircle[fillstyle=solid,fillcolor=fillval](3.295918,0){0.017123}
\pscircle[fillstyle=solid,fillcolor=fillval](3.05102,-0.44898){0.017123}
\pscircle[fillstyle=solid,fillcolor=fillval](3.173469,-0.44898){0.017123}
\pscircle[fillstyle=solid,fillcolor=fillval](3.295918,-0.44898){0.017123}
\psarc(3.112245,-0.010204){0.061224}{180}{360}
\psline(3.295918,-0.010204)(3.295918,-0.438776)
\psarcn(3.112245,-0.438776){0.061224}{180}{0}
\pscircle(2.938776,-0.22449){0.057939}
\rput(2.622449,-0.969388){$U_{1}^2 = \delta U_{1}$}
\end{pspicture}%

%% file: pictures/TLrel3.tex
\psset{unit=1in,cornersize=absolute,dimen=middle}%
\begin{pspicture}(0,-0.789474)(1.642105,0.010526)%
\psset{linewidth=1pt}%
\newgray{fillval}{0.3}
\pscircle[fillstyle=solid,fillcolor=fillval](0.010526,0){0.010526}
\pscircle[fillstyle=solid,fillcolor=fillval](0.136842,0){0.010526}
\pscircle[fillstyle=solid,fillcolor=fillval](0.263158,0){0.010526}
\pscircle[fillstyle=solid,fillcolor=fillval](0.389474,0){0.010526}
\pscircle[fillstyle=solid,fillcolor=fillval](0.010526,-0.231579){0.010526}
\pscircle[fillstyle=solid,fillcolor=fillval](0.136842,-0.231579){0.010526}
\pscircle[fillstyle=solid,fillcolor=fillval](0.263158,-0.231579){0.010526}
\pscircle[fillstyle=solid,fillcolor=fillval](0.389474,-0.231579){0.010526}
\pscircle[fillstyle=solid,fillcolor=fillval](0.010526,-0.463158){0.010526}
\pscircle[fillstyle=solid,fillcolor=fillval](0.136842,-0.463158){0.010526}
\pscircle[fillstyle=solid,fillcolor=fillval](0.263158,-0.463158){0.010526}
\pscircle[fillstyle=solid,fillcolor=fillval](0.389474,-0.463158){0.010526}
\psarc(0.073684,-0.010526){0.063158}{-179.999999}{-0.000001}
\psline(0.263158,-0.010526)(0.263158,-0.221053)
\psline(0.389474,-0.010526)(0.389474,-0.242105)
\psarcn(0.073684,-0.221053){0.063158}{180}{0}
\psarc(0.326316,-0.242105){0.063158}{180}{360}
\psline(0.010526,-0.242105)(0.010526,-0.452632)
\psline(0.136842,-0.242105)(0.136842,-0.452632)
\psarcn(0.326316,-0.452632){0.063158}{180}{0}
\rput(0.821053,-0.231579){{\huge $=$}}
\pscircle[fillstyle=solid,fillcolor=fillval](1.252632,0){0.010526}
\pscircle[fillstyle=solid,fillcolor=fillval](1.378947,0){0.010526}
\pscircle[fillstyle=solid,fillcolor=fillval](1.505263,0){0.010526}
\pscircle[fillstyle=solid,fillcolor=fillval](1.631579,0){0.010526}
\pscircle[fillstyle=solid,fillcolor=fillval](1.252632,-0.231579){0.010526}
\pscircle[fillstyle=solid,fillcolor=fillval](1.378947,-0.231579){0.010526}
\pscircle[fillstyle=solid,fillcolor=fillval](1.505263,-0.231579){0.010526}
\pscircle[fillstyle=solid,fillcolor=fillval](1.631579,-0.231579){0.010526}
\pscircle[fillstyle=solid,fillcolor=fillval](1.252632,-0.463158){0.010526}
\pscircle[fillstyle=solid,fillcolor=fillval](1.378947,-0.463158){0.010526}
\pscircle[fillstyle=solid,fillcolor=fillval](1.505263,-0.463158){0.010526}
\pscircle[fillstyle=solid,fillcolor=fillval](1.631579,-0.463158){0.010526}
\psarc(1.568421,-0.010526){0.063158}{-179.999998}{-0.000002}
\psline(1.252632,-0.010526)(1.252632,-0.221053)
\psline(1.378947,-0.010526)(1.378947,-0.221053)
\psarcn(1.568421,-0.221053){0.063158}{179.999998}{0.000002}
\psarc(1.315789,-0.242105){0.063158}{-179.999998}{-0.000002}
\psline(1.505263,-0.242105)(1.505263,-0.452632)
\psline(1.631579,-0.242105)(1.631579,-0.452632)
\psarcn(1.315789,-0.452632){0.063158}{179.999998}{0.000002}
\rput(0.810526,-0.789474){$U_{1}U_{3} = U_{3} U_{1}$}
\end{pspicture}%

%% file: pictures/TLM3.tex
\psset{unit=1in,cornersize=absolute,dimen=middle}%
\begin{pspicture}(0,-0.542929)(2.5,0.012626)%
\psset{linewidth=1pt}%
\newgray{fillval}{0.3}
\pscircle[fillstyle=solid,fillcolor=fillval](0.012626,0){0.012626}
\pscircle[fillstyle=solid,fillcolor=fillval](0.164141,0){0.012626}
\pscircle[fillstyle=solid,fillcolor=fillval](0.315657,0){0.012626}
\pscircle[fillstyle=solid,fillcolor=fillval](0.012626,-0.530303){0.012626}
\pscircle[fillstyle=solid,fillcolor=fillval](0.164141,-0.530303){0.012626}
\pscircle[fillstyle=solid,fillcolor=fillval](0.315657,-0.530303){0.012626}
\psline(0.012626,-0.012626)(0.012626,-0.517677)
\psline(0.164141,-0.012626)(0.164141,-0.517677)
\psline(0.315657,-0.012626)(0.315657,-0.517677)
\pscircle[fillstyle=solid,fillcolor=fillval](1.098485,0){0.012626}
\pscircle[fillstyle=solid,fillcolor=fillval](1.25,0){0.012626}
\pscircle[fillstyle=solid,fillcolor=fillval](1.401515,0){0.012626}
\pscircle[fillstyle=solid,fillcolor=fillval](1.098485,-0.530303){0.012626}
\pscircle[fillstyle=solid,fillcolor=fillval](1.25,-0.530303){0.012626}
\pscircle[fillstyle=solid,fillcolor=fillval](1.401515,-0.530303){0.012626}
\psarc(1.325758,-0.012626){0.075758}{-179.999998}{-0.000002}
\psarcn(1.174242,-0.517677){0.075758}{179.999998}{0.000002}
\psbezier(1.098485,-0.012626)(1.098485,-0.03367)(1.098485,-0.054714)(1.098485,-0.075758)
(1.098485,-0.117845)(1.14899,-0.180976)(1.25,-0.265152)
(1.35101,-0.349327)(1.401515,-0.412458)(1.401515,-0.454545)
(1.401515,-0.475589)(1.401515,-0.496633)(1.401515,-0.517677)
\pscircle[fillstyle=solid,fillcolor=fillval](2.184343,0){0.012626}
\pscircle[fillstyle=solid,fillcolor=fillval](2.335859,0){0.012626}
\pscircle[fillstyle=solid,fillcolor=fillval](2.487374,0){0.012626}
\pscircle[fillstyle=solid,fillcolor=fillval](2.184343,-0.530303){0.012626}
\pscircle[fillstyle=solid,fillcolor=fillval](2.335859,-0.530303){0.012626}
\pscircle[fillstyle=solid,fillcolor=fillval](2.487374,-0.530303){0.012626}
\psarc(2.260101,-0.012626){0.075758}{180}{360}
\psarcn(2.411616,-0.517677){0.075758}{180}{0}
\psbezier(2.487374,-0.012626)(2.487374,-0.03367)(2.487374,-0.054714)(2.487374,-0.075758)
(2.487374,-0.117845)(2.436869,-0.180976)(2.335859,-0.265152)
(2.234848,-0.349327)(2.184343,-0.412458)(2.184343,-0.454545)
(2.184343,-0.475589)(2.184343,-0.496633)(2.184343,-0.517677)
\end{pspicture}%

%% file: pictures/wave.tex
\psset{unit=1in,cornersize=absolute,dimen=middle}%
\begin{pspicture}(0,-0.733696)(1.826087,0.016304)%
\psset{linewidth=1pt}%
\newgray{fillval}{0.3}
\pscircle[fillstyle=solid,fillcolor=fillval](0.016304,0){0.016304}
\pscircle[fillstyle=solid,fillcolor=fillval](0.211957,0){0.016304}
\pscircle[fillstyle=solid,fillcolor=fillval](0.407609,0){0.016304}
\pscircle[fillstyle=solid,fillcolor=fillval](0.016304,-0.684783){0.016304}
\pscircle[fillstyle=solid,fillcolor=fillval](0.211957,-0.684783){0.016304}
\pscircle[fillstyle=solid,fillcolor=fillval](0.407609,-0.684783){0.016304}
\psarc(0.309783,-0.016304){0.097826}{180}{360}
\psarcn(0.11413,-0.668478){0.097826}{180}{0}
\psbezier(0.016304,-0.016304)(0.016304,-0.043478)(0.016304,-0.070652)(0.016304,-0.097826)
(0.016304,-0.152174)(0.081522,-0.233696)(0.211957,-0.342391)
(0.342391,-0.451087)(0.407609,-0.532609)(0.407609,-0.586957)
(0.407609,-0.61413)(0.407609,-0.641304)(0.407609,-0.668478)
\rput(0.913043,-0.358696){{\huge $=$}}
\pscircle[fillstyle=solid,fillcolor=fillval](1.418478,0){0.016304}
\pscircle[fillstyle=solid,fillcolor=fillval](1.61413,0){0.016304}
\pscircle[fillstyle=solid,fillcolor=fillval](1.809783,0){0.016304}
\pscircle[fillstyle=solid,fillcolor=fillval](1.418478,-0.358696){0.016304}
\pscircle[fillstyle=solid,fillcolor=fillval](1.61413,-0.358696){0.016304}
\pscircle[fillstyle=solid,fillcolor=fillval](1.809783,-0.358696){0.016304}
\pscircle[fillstyle=solid,fillcolor=fillval](1.418478,-0.717391){0.016304}
\pscircle[fillstyle=solid,fillcolor=fillval](1.61413,-0.717391){0.016304}
\pscircle[fillstyle=solid,fillcolor=fillval](1.809783,-0.717391){0.016304}
\psline(1.418478,-0.016304)(1.418478,-0.342391)
\psarc(1.711957,-0.016304){0.097826}{-179.999998}{-0.000002}
\psarcn(1.711957,-0.342391){0.097826}{179.999998}{0.000002}
\psarc(1.516304,-0.375){0.097826}{-179.999998}{-0.000002}
\psarcn(1.516304,-0.701087){0.097826}{179.999998}{0.000002}
\psline(1.809783,-0.375)(1.809783,-0.701087)
\end{pspicture}%

%% file: pictures/nest.tex
\psset{unit=1in,cornersize=absolute,dimen=middle}%
\begin{pspicture}(0,-0.733696)(2.217391,0.016304)%
\psset{linewidth=1pt}%
\newgray{fillval}{0.3}
\pscircle[fillstyle=solid,fillcolor=fillval](0.016304,0){0.016304}
\pscircle[fillstyle=solid,fillcolor=fillval](0.211957,0){0.016304}
\pscircle[fillstyle=solid,fillcolor=fillval](0.407609,0){0.016304}
\pscircle[fillstyle=solid,fillcolor=fillval](0.603261,0){0.016304}
\pscircle[fillstyle=solid,fillcolor=fillval](0.016304,-0.684783){0.016304}
\pscircle[fillstyle=solid,fillcolor=fillval](0.211957,-0.684783){0.016304}
\pscircle[fillstyle=solid,fillcolor=fillval](0.407609,-0.684783){0.016304}
\pscircle[fillstyle=solid,fillcolor=fillval](0.603261,-0.684783){0.016304}
\psarc(0.309783,-0.016304){0.097826}{180}{360}
\psarc(0.309783,-0.016304){0.293478}{180}{360}
\psarcn(0.309783,-0.668478){0.097826}{180}{0}
\psarcn(0.309783,-0.668478){0.293478}{180}{0}
\rput(1.108696,-0.358696){{\huge $=$}}
\pscircle[fillstyle=solid,fillcolor=fillval](1.61413,0){0.016304}
\pscircle[fillstyle=solid,fillcolor=fillval](1.809783,0){0.016304}
\pscircle[fillstyle=solid,fillcolor=fillval](2.005435,0){0.016304}
\pscircle[fillstyle=solid,fillcolor=fillval](2.201087,0){0.016304}
\pscircle[fillstyle=solid,fillcolor=fillval](1.61413,-0.358696){0.016304}
\pscircle[fillstyle=solid,fillcolor=fillval](1.809783,-0.358696){0.016304}
\pscircle[fillstyle=solid,fillcolor=fillval](2.005435,-0.358696){0.016304}
\pscircle[fillstyle=solid,fillcolor=fillval](2.201087,-0.358696){0.016304}
\pscircle[fillstyle=solid,fillcolor=fillval](1.61413,-0.717391){0.016304}
\pscircle[fillstyle=solid,fillcolor=fillval](1.809783,-0.717391){0.016304}
\pscircle[fillstyle=solid,fillcolor=fillval](2.005435,-0.717391){0.016304}
\pscircle[fillstyle=solid,fillcolor=fillval](2.201087,-0.717391){0.016304}
\psbezier(1.61413,-0.016304)(1.61413,-0.029891)(1.61413,-0.043478)(1.61413,-0.057065)
(1.61413,-0.084239)(1.679348,-0.125)(1.809783,-0.179348)
(1.940217,-0.233696)(2.005435,-0.274457)(2.005435,-0.30163)
(2.005435,-0.315217)(2.005435,-0.328804)(2.005435,-0.342391)
\psline(2.201087,-0.016304)(2.201087,-0.342391)
\psarc(1.907609,-0.016304){0.097826}{-179.999998}{-0.000002}
\psarcn(1.711957,-0.342391){0.097826}{179.999998}{0.000002}
\psarc(2.103261,-0.375){0.097826}{-179.999998}{-0.000002}
\psarcn(1.907609,-0.701087){0.097826}{179.999998}{0.000002}
\psline(1.61413,-0.375)(1.61413,-0.701087)
\psbezier(1.809783,-0.375)(1.809783,-0.388587)(1.809783,-0.402174)(1.809783,-0.415761)
(1.809783,-0.442935)(1.875,-0.483696)(2.005435,-0.538043)
(2.13587,-0.592391)(2.201087,-0.633152)(2.201087,-0.660326)
(2.201087,-0.673913)(2.201087,-0.6875)(2.201087,-0.701087)
\end{pspicture}%

%% file: pictures/ear.tex
\psset{unit=1in,cornersize=absolute,dimen=middle}%
\begin{pspicture}(0,-0.819149)(2,0.361702)%
\psset{linewidth=1pt}%
\newgray{fillval}{0.3}
\pscircle[fillstyle=solid,fillcolor=fillval](0.015957,0){0.015957}
\pscircle[fillstyle=solid,fillcolor=fillval](0.154255,0){0.015957}
\pscircle[fillstyle=solid,fillcolor=fillval](0.292553,0){0.015957}
\pscircle[fillstyle=solid,fillcolor=fillval](0.015957,-0.457447){0.015957}
\pscircle[fillstyle=solid,fillcolor=fillval](0.154255,-0.457447){0.015957}
\pscircle[fillstyle=solid,fillcolor=fillval](0.292553,-0.457447){0.015957}
\psarc(0.223404,-0.015957){0.069149}{180}{360}
\psarcn(0.085106,-0.441489){0.069149}{180}{0}
\psbezier(0.015957,-0.015957)(0.015957,-0.033688)(0.015957,-0.051418)(0.015957,-0.069149)
(0.015957,-0.10461)(0.062057,-0.157801)(0.154255,-0.228723)
(0.246454,-0.299645)(0.292553,-0.352837)(0.292553,-0.388298)
(0.292553,-0.406028)(0.292553,-0.423759)(0.292553,-0.441489)
\psarc(0.361702,-0.473404){0.069149}{180}{360}
\psarc(0.361702,-0.473404){0.207447}{180}{360}
\psarc(0.361702,-0.473404){0.345745}{180}{360}
\psarcn(0.361702,0.015957){0.345745}{180}{0}
\psarcn(0.361702,0.015957){0.207447}{180}{0}
\psarcn(0.361702,0.015957){0.069149}{180}{0}
\psline(0.430851,0.015957)(0.430851,-0.473404)
\psline(0.569149,0.015957)(0.569149,-0.473404)
\psline(0.707447,0.015957)(0.707447,-0.473404)
\rput(1.148936,-0.244681){{\huge $=$}}
\pscircle(1.787234,-0.244681){0.212766}
\end{pspicture}%

%% file: pictures/dimen.tex
\psset{unit=1in,cornersize=absolute,dimen=middle}%
\begin{pspicture}(0,-0.746124)(2.5,0.329457)%
\psset{linewidth=1pt}%
\newgray{fillval}{0.3}
\pscircle[fillstyle=solid,fillcolor=fillval](0.014535,0){0.014535}
\pscircle[fillstyle=solid,fillcolor=fillval](0.140504,0){0.014535}
\pscircle[fillstyle=solid,fillcolor=fillval](0.266473,0){0.014535}
\pscircle[fillstyle=solid,fillcolor=fillval](0.014535,-0.416667){0.014535}
\pscircle[fillstyle=solid,fillcolor=fillval](0.140504,-0.416667){0.014535}
\pscircle[fillstyle=solid,fillcolor=fillval](0.266473,-0.416667){0.014535}
\psline(0.014535,-0.014535)(0.014535,-0.402132)
\psline(0.140504,-0.014535)(0.140504,-0.402132)
\psline(0.266473,-0.014535)(0.266473,-0.402132)
\psarc(0.329457,-0.431202){0.062984}{180}{360}
\psarc(0.329457,-0.431202){0.188953}{180}{360}
\psarc(0.329457,-0.431202){0.314922}{180}{360}
\psarcn(0.329457,0.014535){0.314922}{180}{0}
\psarcn(0.329457,0.014535){0.188953}{180}{0}
\psarcn(0.329457,0.014535){0.062984}{180}{0}
\psline(0.392442,0.014535)(0.392442,-0.431202)
\psline(0.518411,0.014535)(0.518411,-0.431202)
\psline(0.64438,0.014535)(0.64438,-0.431202)
\rput(1.046512,-0.222868){{\huge $=$}}
\pscircle(1.579457,-0.222868){0.145349}
\pscircle(1.967054,-0.222868){0.145349}
\pscircle(2.354651,-0.222868){0.145349}
\end{pspicture}%

%% file: pictures/Kbrac.tex
\psset{unit=1in,cornersize=absolute,dimen=middle}%
\begin{pspicture}(0,-0.2)(4,0.2)%
\psset{linewidth=2pt}%
\psline(0,0)(0.1,0.2)
\psline(0,0)(0.1,-0.2)
\psline(0.8,0)(0.7,0.2)
\psline(0.8,0)(0.7,-0.2)
\psset{linewidth=1pt}%
\psline(0.2,0.2)(0.6,-0.2)
\psline(0.2,-0.2)(0.38,-0.02)
\psline(0.6,0.2)(0.42,0.02)
\psset{linewidth=2pt}%
\psline(1.6,0)(1.7,0.2)
\psline(1.6,0)(1.7,-0.2)
\psline(2.4,0)(2.3,0.2)
\psline(2.4,0)(2.3,-0.2)
\psset{linewidth=1pt}%
\psbezier(1.8,0.2)(1.96,0.088)(2.04,0.088)(2.2,0.2)
\psbezier(1.8,-0.2)(1.96,-0.088)(2.04,-0.088)(2.2,-0.2)
\psset{linewidth=2pt}%
\psline(3.2,0)(3.3,0.2)
\psline(3.2,0)(3.3,-0.2)
\psline(4,0)(3.9,0.2)
\psline(4,0)(3.9,-0.2)
\psset{linewidth=1pt}%
\psbezier(3.4,0.2)(3.46,0.08)(3.46,-0.08)(3.4,-0.2)
\psbezier(3.8,0.2)(3.74,0.08)(3.74,-0.08)(3.8,-0.2)
\rput(1.2,0){$=$}
\rput(2.8,0){$+$}
\uput{0.5ex}[l](1.6,0){\llap{$A$}}
\uput{0.5ex}[l](3.2,0){\llap{$B$}}
\end{pspicture}%

%% file: pictures/purecap.tex
\psset{unit=1in,cornersize=absolute,dimen=middle}%
\begin{pspicture}(0,-0.03125)(0.4375,0.21875)%
\psset{linewidth=1pt}%
\newgray{fillval}{0.3}
\pscircle[fillstyle=solid,fillcolor=fillval](0.03125,0){0.03125}
\pscircle[fillstyle=solid,fillcolor=fillval](0.40625,0){0.03125}
\psarcn(0.21875,0.03125){0.1875}{179.999999}{0.000001}
\end{pspicture}%

%% file: pictures/unitn.tex
\psset{unit=1in,cornersize=absolute,dimen=middle}%
\begin{pspicture}(0,-0.025876)(0.948249,0.474124)%
\psset{linewidth=1pt}%
\newgray{fillval}{0.3}
\pscircle[fillstyle=solid,fillcolor=fillval](0.015294,0){0.015294}
\pscircle[fillstyle=solid,fillcolor=fillval](0.382358,0){0.015294}
\pscircle[fillstyle=solid,fillcolor=fillval](0.56589,0){0.015294}
\pscircle[fillstyle=solid,fillcolor=fillval](0.932954,0){0.015294}
\psarcn(0.474124,0.015294){0.091766}{-180}{-360}
\psarcn(0.474124,0.015294){0.45883}{-180}{-360}
\rput(0.198826,0){$\ldots$}
\rput(0.749422,0){$\ldots$}
\uput{0.5ex}[d](0.015294,-0.015294){$1$}
\uput{0.5ex}[d](0.932954,-0.015294){$2n$}
\end{pspicture}%

%% file: pictures/counitn.tex
\psset{unit=1in,cornersize=absolute,dimen=middle}%
\begin{pspicture}(0,-0.474124)(0.948249,0.025876)%
\psset{linewidth=1pt}%
\newgray{fillval}{0.3}
\pscircle[fillstyle=solid,fillcolor=fillval](0.015294,0){0.015294}
\pscircle[fillstyle=solid,fillcolor=fillval](0.382358,0){0.015294}
\pscircle[fillstyle=solid,fillcolor=fillval](0.56589,0){0.015294}
\pscircle[fillstyle=solid,fillcolor=fillval](0.932954,0){0.015294}
\psarc(0.474124,-0.015294){0.091766}{180}{360}
\psarc(0.474124,-0.015294){0.45883}{180}{360}
\rput(0.198826,0){$\ldots$}
\rput(0.749422,0){$\ldots$}
\uput{0.5ex}[u](0.015294,0.015294){$1$}
\uput{0.5ex}[u](0.932954,0.015294){$2n$}
\end{pspicture}%

%% file: pictures/cupcapeqn.tex
\psset{unit=1in,cornersize=absolute,dimen=middle}%
\begin{pspicture}(0,-0.733696)(2.184783,0.016304)%
\psset{linewidth=1pt}%
\newgray{fillval}{0.3}
\pscircle[fillstyle=solid,fillcolor=fillval](0.016304,0){0.016304}
\pscircle[fillstyle=solid,fillcolor=fillval](0.016304,-0.358696){0.016304}
\pscircle[fillstyle=solid,fillcolor=fillval](0.211957,-0.358696){0.016304}
\pscircle[fillstyle=solid,fillcolor=fillval](0.407609,-0.358696){0.016304}
\pscircle[fillstyle=solid,fillcolor=fillval](0.407609,-0.717391){0.016304}
\pscircle[fillstyle=solid,fillcolor=fillval](1.092391,-0.717391){0.016304}
\pscircle[fillstyle=solid,fillcolor=fillval](1.777174,-0.717391){0.016304}
\pscircle[fillstyle=solid,fillcolor=fillval](1.777174,-0.358696){0.016304}
\pscircle[fillstyle=solid,fillcolor=fillval](1.972826,-0.358696){0.016304}
\pscircle[fillstyle=solid,fillcolor=fillval](2.168478,-0.358696){0.016304}
\pscircle[fillstyle=solid,fillcolor=fillval](1.092391,-0){0.016304}
\pscircle[fillstyle=solid,fillcolor=fillval](2.168478,-0){0.016304}
\psarcn(0.309783,-0.342391){0.097826}{180}{0}
\psarcn(1.875,-0.342391){0.097826}{179.999998}{0.000002}
\psarc(0.11413,-0.375){0.097826}{-179.999999}{-0.000001}
\psarc(2.070652,-0.375){0.097826}{-179.999998}{-0.000002}
\psline(0.016304,-0.016304)(0.016304,-0.342391)
\psline(0.407609,-0.375)(0.407609,-0.701087)
\psline(1.092391,-0.016304)(1.092391,-0.701087)
\psline(1.777174,-0.375)(1.777174,-0.701087)
\psline(2.168478,-0.016304)(2.168478,-0.342391)
\rput(0.75,-0.358696){$=$}
\rput(1.402174,-0.358696){$=$}
\end{pspicture}%

%% file: pictures/stareqn.tex
\psset{unit=1in,cornersize=absolute,dimen=middle}%
\begin{pspicture}(-0.625,-1.235795)(1.221591,0.014205)%
\psset{linewidth=1pt}%
\newgray{fillval}{0.3}
\pscircle[fillstyle=solid,fillcolor=fillval](0.014205,0){0.014205}
\pscircle[fillstyle=solid,fillcolor=fillval](0.610795,0){0.014205}
\pscircle[fillstyle=solid,fillcolor=fillval](0.014205,-0.3125){0.014205}
\pscircle[fillstyle=solid,fillcolor=fillval](-0.298295,-0.3125){0.014205}
\pscircle[fillstyle=solid,fillcolor=fillval](-0.610795,-0.3125){0.014205}
\pscircle[fillstyle=solid,fillcolor=fillval](0.610795,-0.3125){0.014205}
\pscircle[fillstyle=solid,fillcolor=fillval](0.923295,-0.3125){0.014205}
\pscircle[fillstyle=solid,fillcolor=fillval](1.207386,-0.3125){0.014205}
\pscircle[fillstyle=solid,fillcolor=fillval](0.014205,-0.909091){0.014205}
\pscircle[fillstyle=solid,fillcolor=fillval](-0.298295,-0.909091){0.014205}
\pscircle[fillstyle=solid,fillcolor=fillval](-0.610795,-0.909091){0.014205}
\pscircle[fillstyle=solid,fillcolor=fillval](0.610795,-0.909091){0.014205}
\pscircle[fillstyle=solid,fillcolor=fillval](0.923295,-0.909091){0.014205}
\pscircle[fillstyle=solid,fillcolor=fillval](1.207386,-0.909091){0.014205}
\pscircle[fillstyle=solid,fillcolor=fillval](-0.610795,-1.221591){0.014205}
\pscircle[fillstyle=solid,fillcolor=fillval](1.207386,-1.221591){0.014205}
\psframe(-0.440341,-0.752841)(-0.15625,-0.46875)
\rput(-0.298295,-0.610795){$f$}
\psframe(0.78125,-0.752841)(1.065341,-0.46875)
\rput(0.923295,-0.610795){$f$}
\psarcn(-0.454545,-0.298295){0.15625}{180}{0}
\psarcn(1.065341,-0.298295){0.142045}{180}{0}
\psarc(-0.142045,-0.923295){0.15625}{180}{360}
\psarc(0.767045,-0.923295){0.15625}{180}{360}
\psline(0.014205,-0.014205)(0.014205,-0.298295)
\psline(-0.610795,-0.326705)(-0.610795,-0.894886)
\psline(-0.610795,-0.923295)(-0.610795,-1.207386)
\psline(-0.298295,-0.326705)(-0.298295,-0.46875)
\psline(-0.298295,-0.752841)(-0.298295,-0.894886)
\psline(0.014205,-0.326705)(0.014205,-0.894886)
\psline(0.610795,-0.014205)(0.610795,-0.298295)
\psline(0.610795,-0.326705)(0.610795,-0.894886)
\psline(0.923295,-0.326705)(0.923295,-0.46875)
\psline(0.923295,-0.752841)(0.923295,-0.894886)
\psline(1.207386,-0.326705)(1.207386,-0.894886)
\psline(1.207386,-0.923295)(1.207386,-1.207386)
\rput(0.3125,-0.610795){$=$}
\end{pspicture}%

%% file: pictures/dualinv.tex
\psset{unit=1in,cornersize=absolute,dimen=middle}%
\begin{pspicture}(0,-0.767296)(3.5,0.767296)%
\psframe[fillstyle=solid,fillcolor=yellow,linecolor=blue](0.75,-0.25)(1.25,0.25)
\rput(1,0){$f$}
\psbezier(1,0.25)(1,0.65)(0.95,0.75)(0.75,0.75)
(0.55,0.75)(0.5,0.65)(0.5,0.25)
(0.5,-0.15)(0.45,-0.25)(0.25,-0.25)
(0.05,-0.25)(0,-0.15)(0,0.25)
\psbezier(1,-0.25)(1,-0.65)(1.05,-0.75)(1.25,-0.75)
(1.45,-0.75)(1.5,-0.65)(1.5,-0.25)
(1.5,0.15)(1.55,0.25)(1.75,0.25)
(1.95,0.25)(2,0.15)(2,-0.25)
\rput(0,0.35){$A$}
\rput(0.9,-0.35){$B$}
\rput(0.9,0.35){$A$}
\rput(2,-0.35){$B$}
\rput(0.4,0.35){$A^*$}
\rput(1.6,-0.35){$B^*$}
\rput(2.5,0){$=$}
\psframe[fillstyle=solid,fillcolor=yellow,linecolor=blue](3,-0.25)(3.5,0.25)
\rput(3.25,0){$f$}
\psline(3.25,0.25)(3.25,0.75)
\uput{0.5ex}[u](3.25,0.75){$A$}
\psline(3.25,-0.25)(3.25,-0.75)
\uput{0.5ex}[d](3.25,-0.75){$B$}
\end{pspicture}%

%% file: pictures/dualbox.tex
\psset{unit=1in,cornersize=absolute,dimen=middle}%
\begin{pspicture}(0,-0.335692)(1.75,0.335692)%
\psframe[fillstyle=solid,fillcolor=yellow,linecolor=blue](0,-0.109375)(0.65625,0.109375)
\rput(0.328125,0){$f$}
\psline(0.164063,0.109375)(0.164063,0.328125)
\psline(0.492188,0.109375)(0.492188,0.328125)
\rput(0.328125,0.328125){$\cdots$}
\uput{0.5ex}[u](0.164063,0.328125){$A_1$}
\uput{0.5ex}[u](0.492188,0.328125){$A_n$}
\psline(0.164063,-0.109375)(0.164063,-0.328125)
\psline(0.492188,-0.109375)(0.492188,-0.328125)
\rput(0.328125,-0.328125){$\cdots$}
\uput{0.5ex}[d](0.164063,-0.328125){$B_1$}
\uput{0.5ex}[d](0.492188,-0.328125){$B_m$}
\psframe[fillstyle=solid,fillcolor=yellow,linecolor=blue](1.09375,-0.109375)(1.75,0.109375)
\rput(1.421875,0){$f^*$}
\psline(1.257813,0.109375)(1.257813,0.328125)
\psline(1.585938,0.109375)(1.585938,0.328125)
\rput(1.421875,0.328125){$\cdots$}
\uput{0.5ex}[u](1.257813,0.328125){$B_{m}^{*}$}
\uput{0.5ex}[u](1.585938,0.328125){$B_{1}^{*}$}
\psline(1.257813,-0.109375)(1.257813,-0.328125)
\psline(1.585938,-0.109375)(1.585938,-0.328125)
\rput(1.421875,-0.328125){$\cdots$}
\uput{0.5ex}[d](1.257813,-0.328125){$A_n^*$}
\uput{0.5ex}[d](1.585938,-0.328125){$A_1^*$}
\rput(0.875,0){$\rightsquigarrow$}
\end{pspicture}%

%% file: pictures/daggerbox.tex
\psset{unit=1in,cornersize=absolute,dimen=middle}%
\begin{pspicture}(0,-0.335692)(1.75,0.335692)%
\psframe[fillstyle=solid,fillcolor=yellow,linecolor=blue](0,-0.109375)(0.65625,0.109375)
\rput(0.328125,0){$f$}
\psline(0.164063,0.109375)(0.164063,0.328125)
\psline(0.492188,0.109375)(0.492188,0.328125)
\rput(0.328125,0.328125){$\cdots$}
\uput{0.5ex}[u](0.164063,0.328125){$A_1$}
\uput{0.5ex}[u](0.492188,0.328125){$A_n$}
\psline(0.164063,-0.109375)(0.164063,-0.328125)
\psline(0.492188,-0.109375)(0.492188,-0.328125)
\rput(0.328125,-0.328125){$\cdots$}
\uput{0.5ex}[d](0.164063,-0.328125){$B_1$}
\uput{0.5ex}[d](0.492188,-0.328125){$B_m$}
\psframe[fillstyle=solid,fillcolor=yellow,linecolor=blue](1.09375,-0.109375)(1.75,0.109375)
\rput(1.421875,0){$f^{\dagger}$}
\psline(1.257813,0.109375)(1.257813,0.328125)
\psline(1.585938,0.109375)(1.585938,0.328125)
\rput(1.421875,0.328125){$\cdots$}
\uput{0.5ex}[u](1.257813,0.328125){$B_{1}$}
\uput{0.5ex}[u](1.585938,0.328125){$B_{n}$}
\psline(1.257813,-0.109375)(1.257813,-0.328125)
\psline(1.585938,-0.109375)(1.585938,-0.328125)
\rput(1.421875,-0.328125){$\cdots$}
\uput{0.5ex}[d](1.257813,-0.328125){$A_1$}
\uput{0.5ex}[d](1.585938,-0.328125){$A_n$}
\rput(0.875,0){$\rightsquigarrow$}
\end{pspicture}%

%% file: pictures/LRwave.tex
\psset{unit=1in,cornersize=absolute,dimen=middle}%
\begin{pspicture}(0,-0.767857)(2,0.017857)%
\psset{linewidth=1pt}%
\newgray{fillval}{0.3}
\pscircle[fillstyle=solid,fillcolor=fillval](0.017857,0){0.017857}
\pscircle[fillstyle=solid,fillcolor=fillval](0.232143,0){0.017857}
\pscircle[fillstyle=solid,fillcolor=fillval](0.446429,0){0.017857}
\pscircle[fillstyle=solid,fillcolor=fillval](0.017857,-0.75){0.017857}
\pscircle[fillstyle=solid,fillcolor=fillval](0.232143,-0.75){0.017857}
\pscircle[fillstyle=solid,fillcolor=fillval](0.446429,-0.75){0.017857}
\psarc(0.339286,-0.017857){0.107143}{180}{360}
\psarcn(0.125,-0.732143){0.107143}{180}{0}
\psbezier(0.017857,-0.017857)(0.017857,-0.047619)(0.017857,-0.077381)(0.017857,-0.107143)
(0.017857,-0.166667)(0.089286,-0.255952)(0.232143,-0.375)
(0.375,-0.494048)(0.446429,-0.583333)(0.446429,-0.642857)
(0.446429,-0.672619)(0.446429,-0.702381)(0.446429,-0.732143)
\pscircle[fillstyle=solid,fillcolor=fillval](1.553571,0){0.017857}
\pscircle[fillstyle=solid,fillcolor=fillval](1.767857,0){0.017857}
\pscircle[fillstyle=solid,fillcolor=fillval](1.982143,0){0.017857}
\pscircle[fillstyle=solid,fillcolor=fillval](1.553571,-0.75){0.017857}
\pscircle[fillstyle=solid,fillcolor=fillval](1.767857,-0.75){0.017857}
\pscircle[fillstyle=solid,fillcolor=fillval](1.982143,-0.75){0.017857}
\psarc(1.660714,-0.017857){0.107143}{-179.999998}{-0.000002}
\psarcn(1.875,-0.732143){0.107143}{179.999998}{0.000002}
\psbezier(1.982143,-0.017857)(1.982143,-0.047619)(1.982143,-0.077381)(1.982143,-0.107143)
(1.982143,-0.166667)(1.910714,-0.255952)(1.767857,-0.375)
(1.625,-0.494048)(1.553571,-0.583333)(1.553571,-0.642857)
(1.553571,-0.672619)(1.553571,-0.702381)(1.553571,-0.732143)
\end{pspicture}%

%% file: pictures/fourex.tex
\psset{unit=1in,cornersize=absolute,dimen=middle}%
\begin{pspicture}(0,-0.317302)(3,0.013393)%
\psset{linewidth=1pt}%
\newgray{fillval}{0.3}
\pscircle[fillstyle=solid,fillcolor=fillval](0.013393,0){0.013393}
\pscircle[fillstyle=solid,fillcolor=fillval](0.174107,0){0.013393}
\pscircle[fillstyle=solid,fillcolor=fillval](0.334821,0){0.013393}
\pscircle[fillstyle=solid,fillcolor=fillval](0.897321,0){0.013393}
\pscircle[fillstyle=solid,fillcolor=fillval](0.334821,-0.294643){0.013393}
\pscircle[fillstyle=solid,fillcolor=fillval](0.897321,-0.294643){0.013393}
\pscircle[fillstyle=solid,fillcolor=fillval](1.058036,-0.294643){0.013393}
\pscircle[fillstyle=solid,fillcolor=fillval](1.21875,-0.294643){0.013393}
\pscircle[fillstyle=solid,fillcolor=fillval](1.78125,-0.294643){0.013393}
\pscircle[fillstyle=solid,fillcolor=fillval](1.941964,-0.294643){0.013393}
\pscircle[fillstyle=solid,fillcolor=fillval](2.102679,-0.294643){0.013393}
\pscircle[fillstyle=solid,fillcolor=fillval](2.102679,0){0.013393}
\pscircle[fillstyle=solid,fillcolor=fillval](2.665179,0){0.013393}
\pscircle[fillstyle=solid,fillcolor=fillval](2.825893,0){0.013393}
\pscircle[fillstyle=solid,fillcolor=fillval](2.986607,0){0.013393}
\pscircle[fillstyle=solid,fillcolor=fillval](2.665179,-0.294643){0.013393}
\psarc(0.09375,-0.013393){0.080357}{-179.999999}{-0.000001}
\psarc(2.90625,-0.013393){0.080357}{180}{360}
\psarcn(1.138393,-0.28125){0.080357}{179.999998}{0.000002}
\psarcn(1.861607,-0.28125){0.080357}{179.999998}{0.000002}
\psline(0.334821,-0.013393)(0.334821,-0.28125)
\psline(0.897321,-0.013393)(0.897321,-0.28125)
\psline(2.102679,-0.013393)(2.102679,-0.28125)
\psline(2.665179,-0.013393)(2.665179,-0.28125)
\uput{0.5ex}[d](0.174107,-0.308036){\raisebox{-3ex}{$f$}}
\uput{0.5ex}[d](1.058036,-0.308036){\raisebox{-3ex}{$f^*$}}
\uput{0.5ex}[d](1.941964,-0.308036){\raisebox{-3ex}{$f^{\dagger}$}}
\uput{0.5ex}[d](2.825893,-0.308036){\raisebox{-3ex}{$f_{*}$}}
\end{pspicture}%

%% file: pictures/emfact.tex
\psset{unit=1in,cornersize=absolute,dimen=middle}%
\begin{pspicture}(0,-0.767857)(2,0.017857)%
\psset{linewidth=1pt}%
\newgray{fillval}{0.3}
\pscircle[fillstyle=solid,fillcolor=fillval](0.017857,0){0.017857}
\pscircle[fillstyle=solid,fillcolor=fillval](0.232143,0){0.017857}
\pscircle[fillstyle=solid,fillcolor=fillval](0.446429,0){0.017857}
\pscircle[fillstyle=solid,fillcolor=fillval](0.017857,-0.75){0.017857}
\pscircle[fillstyle=solid,fillcolor=fillval](0.232143,-0.75){0.017857}
\pscircle[fillstyle=solid,fillcolor=fillval](0.446429,-0.75){0.017857}
\psarc(0.339286,-0.017857){0.107143}{180}{360}
\psarcn(0.125,-0.732143){0.107143}{180}{0}
\psbezier(0.017857,-0.017857)(0.017857,-0.047619)(0.017857,-0.077381)(0.017857,-0.107143)
(0.017857,-0.166667)(0.089286,-0.255952)(0.232143,-0.375)
(0.375,-0.494048)(0.446429,-0.583333)(0.446429,-0.642857)
(0.446429,-0.672619)(0.446429,-0.702381)(0.446429,-0.732143)
\pscircle[fillstyle=solid,fillcolor=fillval](1.553571,0){0.017857}
\pscircle[fillstyle=solid,fillcolor=fillval](1.767857,0){0.017857}
\pscircle[fillstyle=solid,fillcolor=fillval](1.982143,0){0.017857}
\pscircle[fillstyle=solid,fillcolor=fillval](1.553571,-0.392857){0.017857}
\pscircle[fillstyle=solid,fillcolor=fillval](1.553571,-0.75){0.017857}
\pscircle[fillstyle=solid,fillcolor=fillval](1.767857,-0.75){0.017857}
\pscircle[fillstyle=solid,fillcolor=fillval](1.982143,-0.75){0.017857}
\psarc(1.875,-0.017857){0.107143}{-179.999998}{-0.000002}
\psarcn(1.660714,-0.732143){0.107143}{179.999998}{0.000002}
\psline(1.553571,-0.017857)(1.553571,-0.375)
\psbezier(1.553571,-0.410714)(1.553571,-0.425595)(1.553571,-0.440476)(1.553571,-0.455357)
(1.553571,-0.485119)(1.625,-0.52381)(1.767857,-0.571429)
(1.910714,-0.619048)(1.982143,-0.657738)(1.982143,-0.6875)
(1.982143,-0.702381)(1.982143,-0.717262)(1.982143,-0.732143)
\rput(1,-0.392857){$=$}
\end{pspicture}%

%% file: pictures/Nnmp.tex
\psset{unit=1in,cornersize=absolute}%
\begin{pspicture}(0,-0.962721)(1.101717,0.037279)
\newgray{fillval}{0.3}
\pscircle[fillstyle=solid,fillcolor=fillval](0.022034,0){0.027569}
\pscircle[fillstyle=solid,fillcolor=fillval](0.286446,0){0.027569}
\rput(0.550859,0){$\cdots$}
\pscircle[fillstyle=solid,fillcolor=fillval](0.815271,0){0.027569}
\pscircle[fillstyle=solid,fillcolor=fillval](0.022034,-0.925442){0.027569}
\pscircle[fillstyle=solid,fillcolor=fillval](0.286446,-0.925442){0.027569}
\rput(0.550859,-0.925442){$\cdots$}
\rput(0.815271,-0.925442){$\cdots$}
\pscircle[fillstyle=solid,fillcolor=fillval](1.079683,-0.925442){0.027569}
\uput{0.5ex}[u](0.022034,0.022034){$1$}
\uput{0.5ex}[u](0.286446,0.022034){$2$}
\uput{0.5ex}[u](0.815271,0.022034){$n$}
\uput{0.5ex}[d](0.022034,-0.947477){$1'$}
\uput{0.5ex}[d](0.286446,-0.947477){$2'$}
\uput{0.5ex}[d](1.079683,-0.947477){$m'$}
\end{pspicture}%

%% file: pictures/TLexp.tex
\psset{unit=1in,cornersize=absolute}%
\begin{pspicture}(0,-0.962721)(0.837305,0.037279)
\newgray{fillval}{0.3}
\pscircle[fillstyle=solid,fillcolor=fillval](0.022034,0){0.027569}
\pscircle[fillstyle=solid,fillcolor=fillval](0.286446,0){0.027569}
\pscircle[fillstyle=solid,fillcolor=fillval](0.550859,0){0.027569}
\pscircle[fillstyle=solid,fillcolor=fillval](0.815271,0){0.027569}
\pscircle[fillstyle=solid,fillcolor=fillval](0.022034,-0.925442){0.027569}
\pscircle[fillstyle=solid,fillcolor=fillval](0.286446,-0.925442){0.027569}
\uput{0.5ex}[u](0.022034,0.022034){$1$}
\uput{0.5ex}[u](0.286446,0.022034){$2$}
\uput{0.5ex}[u](0.550859,0.022034){$3$}
\uput{0.5ex}[u](0.815271,0.022034){$4$}
\uput{0.5ex}[d](0.022034,-0.947477){$1'$}
\uput{0.5ex}[d](0.286446,-0.947477){$2'$}
\psline(0.022034,-0.022034)(0.286446,-0.903408)
\psline(0.550859,-0.022034)(0.022034,-0.903408)
\psarc(0.550859,-0.022034){0.264412}{180}{360}
\end{pspicture}%

%% file: pictures/PL1exap.tex
\psset{unit=1in,cornersize=absolute,dimen=middle}%
\begin{pspicture}(0,-0.962721)(0.837305,0.037279)%
\rput(0.022034,0){$\cdots$}
\newgray{fillval}{0.3}
\pscircle[fillstyle=solid,fillcolor=fillval](0.286446,0){0.022034}
\pscircle[fillstyle=solid,fillcolor=fillval](0.550859,0){0.022034}
\pscircle[fillstyle=solid,fillcolor=fillval](0.815271,0){0.022034}
\pscircle[fillstyle=solid,fillcolor=fillval](0.022034,-0.925442){0.022034}
\rput(0.286446,-0.925442){$\cdots$}
\uput{0.5ex}[u](0.286446,0.022034){$i$}
\uput{0.5ex}[u](0.550859,0.022034){$j$}
\uput{0.5ex}[u](0.815271,0.022034){$f(i)$}
\uput{0.5ex}[d](0.022034,-0.947477){$f(j)$}
\psline(0.550859,-0.022034)(0.022034,-0.903408)
\psarc(0.550859,-0.022034){0.264412}{180}{360}
\end{pspicture}%

%% file: pictures/PL1exbp.tex
\psset{unit=1in,cornersize=absolute}%
\begin{pspicture}(0,-0.309696)(0.905264,0.040304)
\newgray{fillval}{0.3}
\pscircle[fillstyle=solid,fillcolor=fillval](0.023823,0){0.029358}
\pscircle[fillstyle=solid,fillcolor=fillval](0.309696,0){0.029358}
\pscircle[fillstyle=solid,fillcolor=fillval](0.595568,0){0.029358}
\pscircle[fillstyle=solid,fillcolor=fillval](0.881441,0){0.029358}
\uput{0.5ex}[u](0.023823,0.023823){$i$}
\uput{0.5ex}[u](0.309696,0.023823){$j$}
\uput{0.5ex}[u](0.595568,0.023823){$f(i)$}
\uput{0.5ex}[u](0.881441,0.023823){$f(j)$}
\psarc(0.595568,-0.023823){0.285873}{180}{360}
\psarc(0.309696,-0.023823){0.285873}{180}{360}
\end{pspicture}%

%% file: pictures/PL2exp.tex
\psset{unit=1in,cornersize=absolute}%
\begin{pspicture}(0,-0.962721)(0.837305,0.037279)
\rput(0.022034,0){$\cdots$}
\newgray{fillval}{0.3}
\pscircle[fillstyle=solid,fillcolor=fillval](0.286446,0){0.027569}
\rput(0.550859,0){$\cdots$}
\pscircle[fillstyle=solid,fillcolor=fillval](0.815271,0){0.027569}
\pscircle[fillstyle=solid,fillcolor=fillval](0.022034,-0.925442){0.027569}
\rput(0.286446,-0.925442){$\cdots$}
\pscircle[fillstyle=solid,fillcolor=fillval](0.550859,-0.925442){0.027569}
\uput{0.5ex}[u](0.286446,0.022034){$i$}
\uput{0.5ex}[u](0.815271,0.022034){$j$}
\uput{0.5ex}[d](0.022034,-0.947477){$f(j)$}
\uput{0.5ex}[d](0.550859,-0.947477){$f(i)$}
\psline(0.286446,-0.022034)(0.550859,-0.903408)
\psline(0.815271,-0.022034)(0.022034,-0.903408)
\end{pspicture}%

%% file: pictures/paren.tex
\psset{unit=1in,cornersize=absolute,dimen=middle}%
\begin{pspicture}(0,-0.0125)(0.775,0.2375)%
\psset{linewidth=1pt}%
\newgray{fillval}{0.3}
\pscircle[fillstyle=solid,fillcolor=fillval](0.0125,0){0.0125}
\pscircle[fillstyle=solid,fillcolor=fillval](0.1625,0){0.0125}
\pscircle[fillstyle=solid,fillcolor=fillval](0.3125,0){0.0125}
\pscircle[fillstyle=solid,fillcolor=fillval](0.4625,0){0.0125}
\pscircle[fillstyle=solid,fillcolor=fillval](0.6125,0){0.0125}
\pscircle[fillstyle=solid,fillcolor=fillval](0.7625,0){0.0125}
\psarcn(0.5375,0.0125){0.225}{-180}{-360}
\psarcn(0.0875,0.0125){0.075}{179.999999}{0.000001}
\psarcn(0.5375,0.0125){0.075}{179.999999}{0.000001}
\end{pspicture}%

%% file: pictures/name.tex
\psset{unit=1in,cornersize=absolute,dimen=middle}%
\begin{pspicture}(-0,-0.526961)(2.5,0.379902)%
\psset{linewidth=1pt}%
\newgray{fillval}{0.3}
\pscircle[fillstyle=solid,fillcolor=fillval](0.012255,0){0.012255}
\pscircle[fillstyle=solid,fillcolor=fillval](0.159314,0){0.012255}
\pscircle[fillstyle=solid,fillcolor=fillval](0.306373,0){0.012255}
\pscircle[fillstyle=solid,fillcolor=fillval](0.453431,0){0.012255}
\pscircle[fillstyle=solid,fillcolor=fillval](0.60049,0){0.012255}
\pscircle[fillstyle=solid,fillcolor=fillval](0.747549,0){0.012255}
\pscircle[fillstyle=solid,fillcolor=fillval](0.012255,-0.514706){0.012255}
\pscircle[fillstyle=solid,fillcolor=fillval](0.159314,-0.514706){0.012255}
\pscircle[fillstyle=solid,fillcolor=fillval](0.306373,-0.514706){0.012255}
\pscircle[fillstyle=solid,fillcolor=fillval](0.453431,-0.514706){0.012255}
\pscircle[fillstyle=solid,fillcolor=fillval](0.60049,-0.514706){0.012255}
\pscircle[fillstyle=solid,fillcolor=fillval](0.747549,-0.514706){0.012255}
\psline(0.012255,-0.012255)(0.012255,-0.502451)
\psline(0.159314,-0.012255)(0.159314,-0.502451)
\psline(0.306373,-0.012255)(0.306373,-0.502451)
\psarcn(0.379902,0.012255){0.073529}{-180}{-360}
\psarcn(0.379902,0.012255){0.220588}{-180}{-360}
\psarcn(0.379902,0.012255){0.367647}{-180}{-360}
\psarc(0.67402,-0.012255){0.073529}{-179.999998}{-0.000002}
\psarcn(0.526961,-0.502451){0.073529}{179.999999}{0.000001}
\psbezier(0.453431,-0.012255)(0.453431,-0.03268)(0.453431,-0.053105)(0.453431,-0.073529)
(0.453431,-0.114379)(0.502451,-0.175654)(0.60049,-0.257353)
(0.698529,-0.339052)(0.747549,-0.400327)(0.747549,-0.441176)
(0.747549,-0.461601)(0.747549,-0.482026)(0.747549,-0.502451)
\rput(1.25,-0.269608){$=$}
\pscircle[fillstyle=solid,fillcolor=fillval](1.752451,-0.514706){0.012255}
\pscircle[fillstyle=solid,fillcolor=fillval](1.89951,-0.514706){0.012255}
\pscircle[fillstyle=solid,fillcolor=fillval](2.046569,-0.514706){0.012255}
\pscircle[fillstyle=solid,fillcolor=fillval](2.193627,-0.514706){0.012255}
\pscircle[fillstyle=solid,fillcolor=fillval](2.340686,-0.514706){0.012255}
\pscircle[fillstyle=solid,fillcolor=fillval](2.487745,-0.514706){0.012255}
\psarcn(2.267157,-0.502451){0.220588}{180}{0}
\psarcn(1.82598,-0.502451){0.073529}{179.999998}{0.000002}
\psarcn(2.267157,-0.502451){0.073529}{180}{0}
\end{pspicture}%

%% file: pictures/unname.tex
\psset{unit=1in,cornersize=absolute,dimen=middle}%
\begin{pspicture}(0,-0.64951)(2.5,0.012255)%
\psset{linewidth=1pt}%
\newgray{fillval}{0.3}
\pscircle[fillstyle=solid,fillcolor=fillval](0.012255,0){0.012255}
\pscircle[fillstyle=solid,fillcolor=fillval](0.159314,0){0.012255}
\pscircle[fillstyle=solid,fillcolor=fillval](0.306373,0){0.012255}
\pscircle[fillstyle=solid,fillcolor=fillval](2.120098,0){0.012255}
\pscircle[fillstyle=solid,fillcolor=fillval](2.267157,0){0.012255}
\pscircle[fillstyle=solid,fillcolor=fillval](2.414216,0){0.012255}
\pscircle[fillstyle=solid,fillcolor=fillval](0.012255,-0.269608){0.012255}
\pscircle[fillstyle=solid,fillcolor=fillval](0.159314,-0.269608){0.012255}
\pscircle[fillstyle=solid,fillcolor=fillval](0.306373,-0.269608){0.012255}
\pscircle[fillstyle=solid,fillcolor=fillval](0.453431,-0.269608){0.012255}
\pscircle[fillstyle=solid,fillcolor=fillval](0.60049,-0.269608){0.012255}
\pscircle[fillstyle=solid,fillcolor=fillval](0.747549,-0.269608){0.012255}
\pscircle[fillstyle=solid,fillcolor=fillval](0.894608,-0.269608){0.012255}
\pscircle[fillstyle=solid,fillcolor=fillval](1.041667,-0.269608){0.012255}
\pscircle[fillstyle=solid,fillcolor=fillval](1.188725,-0.269608){0.012255}
\pscircle[fillstyle=solid,fillcolor=fillval](0.894608,-0.539216){0.012255}
\pscircle[fillstyle=solid,fillcolor=fillval](1.041667,-0.539216){0.012255}
\pscircle[fillstyle=solid,fillcolor=fillval](1.188725,-0.539216){0.012255}
\pscircle[fillstyle=solid,fillcolor=fillval](2.193627,-0.539216){0.012255}
\pscircle[fillstyle=solid,fillcolor=fillval](2.340686,-0.539216){0.012255}
\pscircle[fillstyle=solid,fillcolor=fillval](2.487745,-0.539216){0.012255}
\psline(0.012255,-0.012255)(0.012255,-0.257353)
\psline(0.159314,-0.012255)(0.159314,-0.257353)
\psline(0.306373,-0.012255)(0.306373,-0.257353)
\psline(0.894608,-0.281863)(0.894608,-0.526961)
\psline(1.041667,-0.281863)(1.041667,-0.526961)
\psline(1.188725,-0.281863)(1.188725,-0.526961)
\psarcn(0.526961,-0.257353){0.073529}{180}{0}
\psarcn(0.968137,-0.257353){0.220588}{180}{0}
\psarcn(0.968137,-0.257353){0.073529}{179.999998}{0.000002}
\psarc(0.379902,-0.281863){0.367647}{180}{360}
\psarc(0.379902,-0.281863){0.220588}{180}{360}
\psarc(0.379902,-0.281863){0.073529}{180}{360}
\psarc(2.340686,-0.012255){0.073529}{180}{360}
\psarcn(2.267157,-0.526961){0.073529}{180}{0}
\psbezier(2.120098,-0.012255)(2.120098,-0.03268)(2.120098,-0.053105)(2.120098,-0.073529)
(2.120098,-0.114379)(2.181373,-0.179739)(2.303922,-0.269608)
(2.426471,-0.359477)(2.487745,-0.424837)(2.487745,-0.465686)
(2.487745,-0.486111)(2.487745,-0.506536)(2.487745,-0.526961)
\rput(1.691176,-0.269608){$=$}
\end{pspicture}%

%% file: pictures/cupcomp.tex
\psset{unit=1in,cornersize=absolute,dimen=middle}%
\begin{pspicture}(0,-0.674419)(2,0.023256)%
\psset{linewidth=1pt}%
\newgray{fillval}{0.3}
\pscircle[fillstyle=solid,fillcolor=fillval](0.023256,0){0.023256}
\pscircle[fillstyle=solid,fillcolor=fillval](0.023256,-0.511628){0.023256}
\pscircle[fillstyle=solid,fillcolor=fillval](0.302326,-0.511628){0.023256}
\pscircle[fillstyle=solid,fillcolor=fillval](0.581395,-0.511628){0.023256}
\pscircle[fillstyle=solid,fillcolor=fillval](0.860465,-0.511628){0.023256}
\pscircle[fillstyle=solid,fillcolor=fillval](1.418605,-0.511628){0.023256}
\pscircle[fillstyle=solid,fillcolor=fillval](1.697674,-0.511628){0.023256}
\pscircle[fillstyle=solid,fillcolor=fillval](1.976744,-0.511628){0.023256}
\pscircle[fillstyle=solid,fillcolor=fillval](1.976744,0){0.023256}
\psarcn(0.44186,-0.488372){0.139535}{180}{0}
\psarcn(1.55814,-0.488372){0.139535}{179.999998}{0.000002}
\psarc(0.162791,-0.534884){0.139535}{-179.999999}{-0.000001}
\psarc(0.72093,-0.534884){0.139535}{180}{360}
\psarc(1.837209,-0.534884){0.139535}{-179.999998}{-0.000002}
\psline(0.023256,-0.023256)(0.023256,-0.488372)
\psline(1.976744,-0.023256)(1.976744,-0.488372)
\rput(1.139535,-0.511628){$\cdots$}
\end{pspicture}%

%% file: pictures/throughcomp.tex
\psset{unit=1in,cornersize=absolute,dimen=middle}%
\begin{pspicture}(0,-1.046512)(2,0.023256)%
\psset{linewidth=1pt}%
\newgray{fillval}{0.3}
\pscircle[fillstyle=solid,fillcolor=fillval](0.023256,0){0.023256}
\pscircle[fillstyle=solid,fillcolor=fillval](0.023256,-0.511628){0.023256}
\pscircle[fillstyle=solid,fillcolor=fillval](0.302326,-0.511628){0.023256}
\pscircle[fillstyle=solid,fillcolor=fillval](0.581395,-0.511628){0.023256}
\pscircle[fillstyle=solid,fillcolor=fillval](0.860465,-0.511628){0.023256}
\pscircle[fillstyle=solid,fillcolor=fillval](1.418605,-0.511628){0.023256}
\pscircle[fillstyle=solid,fillcolor=fillval](1.697674,-0.511628){0.023256}
\pscircle[fillstyle=solid,fillcolor=fillval](1.976744,-0.511628){0.023256}
\pscircle[fillstyle=solid,fillcolor=fillval](1.976744,-1.023256){0.023256}
\psarcn(0.44186,-0.488372){0.139535}{180}{0}
\psarcn(1.837209,-0.488372){0.139535}{179.999998}{0.000002}
\psarc(0.162791,-0.534884){0.139535}{-179.999999}{-0.000001}
\psarc(0.72093,-0.534884){0.139535}{180}{360}
\psarc(1.55814,-0.534884){0.139535}{-179.999998}{-0.000002}
\psline(0.023256,-0.023256)(0.023256,-0.488372)
\psline(1.976744,-1)(1.976744,-0.534884)
\rput(1.139535,-0.511628){$\cdots$}
\end{pspicture}%

%% file: pictures/cyclic.tex
\psset{unit=1in,cornersize=absolute,dimen=middle}%
\begin{pspicture}(0,-1)(2,0.142857)%
\psset{linewidth=1pt}%
\newgray{fillval}{0.3}
\pscircle[fillstyle=solid,fillcolor=fillval](0.020408,0){0.020408}
\pscircle[fillstyle=solid,fillcolor=fillval](0.265306,0){0.020408}
\pscircle[fillstyle=solid,fillcolor=fillval](0.510204,0){0.020408}
\pscircle[fillstyle=solid,fillcolor=fillval](0.755102,0){0.020408}
\pscircle[fillstyle=solid,fillcolor=fillval](1.244898,0){0.020408}
\pscircle[fillstyle=solid,fillcolor=fillval](1.489796,0){0.020408}
\pscircle[fillstyle=solid,fillcolor=fillval](1.734694,0){0.020408}
\pscircle[fillstyle=solid,fillcolor=fillval](1.979592,0){0.020408}
\psarcn(0.142857,0.020408){0.122449}{179.999999}{0.000001}
\psarcn(1.367347,0.020408){0.122449}{179.999998}{0.000002}
\psarcn(1.857143,0.020408){0.122449}{179.999998}{0.000002}
\psarc(0.387755,-0.020408){0.122449}{180}{360}
\psarcn(0.632653,0.020408){0.122449}{-180}{-360}
\psarc(1.612245,-0.020408){0.122449}{-179.999998}{-0.000002}
\psarc(1,-0.020408){0.979592}{180}{360}
\rput(1,0){$\cdots$}
\end{pspicture}%

%% file: pictures/TLid.tex
\psset{unit=1in,cornersize=absolute,dimen=middle}%
\begin{pspicture}(0,-0.722041)(0.627979,0.027959)%
\newgray{fillval}{0.3}
\pscircle[fillstyle=solid,fillcolor=fillval](0.016526,0){0.016526}
\pscircle[fillstyle=solid,fillcolor=fillval](0.214835,0){0.016526}
\rput(0.413144,0){$\cdots$}
\pscircle[fillstyle=solid,fillcolor=fillval](0.611453,0){0.016526}
\pscircle[fillstyle=solid,fillcolor=fillval](0.016526,-0.694082){0.016526}
\pscircle[fillstyle=solid,fillcolor=fillval](0.214835,-0.694082){0.016526}
\rput(0.413144,-0.694082){$\cdots$}
\pscircle[fillstyle=solid,fillcolor=fillval](0.611453,-0.694082){0.016526}
\uput{0.5ex}[u](0.016526,0.016526){$1$}
\uput{0.5ex}[u](0.214835,0.016526){$2$}
\uput{0.5ex}[u](0.611453,0.016526){$n$}
\uput{0.5ex}[d](0.016526,-0.710608){$1'$}
\uput{0.5ex}[d](0.214835,-0.710608){$2'$}
\uput{0.5ex}[d](0.611453,-0.710608){$n'$}
\psline(0.016526,-0.016526)(0.016526,-0.677556)
\psline(0.214835,-0.016526)(0.214835,-0.677556)
\psline(0.611453,-0.016526)(0.611453,-0.677556)
\end{pspicture}%

%% file: pictures/Bopen.tex
\psset{unit=1in,cornersize=absolute,dimen=middle}%
\begin{pspicture}(-0,-0.873837)(1,0.033837)%
\psset{linewidth=1pt}%
\newgray{fillval}{0.3}
\pscircle[fillstyle=solid,fillcolor=fillval](0.02,0){0.02}
\pscircle[fillstyle=solid,fillcolor=fillval](0.26,0){0.02}
\pscircle[fillstyle=solid,fillcolor=fillval](0.5,0){0.02}
\pscircle[fillstyle=solid,fillcolor=fillval](0.74,0){0.02}
\pscircle[fillstyle=solid,fillcolor=fillval](0.98,0){0.02}
\pscircle[fillstyle=solid,fillcolor=fillval](0.02,-0.84){0.02}
\psarc(0.38,-0.02){0.12}{180}{360}
\psarc(0.86,-0.02){0.12}{-179.999999}{-0.000001}
\psline(0.02,-0.02)(0.02,-0.82)
\uput{0.5ex}[u](0.02,0.02){$x^+$}
\uput{0.5ex}[u](0.26,0.02){$x^-$}
\uput{0.5ex}[u](0.5,0.02){$y^+$}
\uput{0.5ex}[u](0.74,0.02){$y^-$}
\uput{0.5ex}[u](0.98,0.02){$z^+$}
\uput{0.5ex}[d](0.02,-0.86){$o$}
\end{pspicture}%

%% file: pictures/Bclosed.tex
\psset{unit=1in,cornersize=absolute,dimen=middle}%
\begin{pspicture}(0,-0.06761)(1,0.5)%
\psset{linewidth=1pt}%
\newgray{fillval}{0.3}
\pscircle[fillstyle=solid,fillcolor=fillval](0.016129,0){0.016129}
\pscircle[fillstyle=solid,fillcolor=fillval](0.209677,0){0.016129}
\pscircle[fillstyle=solid,fillcolor=fillval](0.403226,0){0.016129}
\pscircle[fillstyle=solid,fillcolor=fillval](0.596774,0){0.016129}
\pscircle[fillstyle=solid,fillcolor=fillval](0.790323,0){0.016129}
\pscircle[fillstyle=solid,fillcolor=fillval](0.983871,0){0.016129}
\psarcn(0.306452,0.016129){0.096774}{-180}{-360}
\psarcn(0.693548,0.016129){0.096774}{179.999999}{0.000001}
\psarcn(0.5,0.016129){0.483871}{-180}{-360}
\uput{0.5ex}[d](0.983871,-0.016129){$x^+$}
\uput{0.5ex}[d](0.790323,-0.016129){$x^-$}
\uput{0.5ex}[d](0.596774,-0.016129){$y^+$}
\uput{0.5ex}[d](0.403226,-0.016129){$y^-$}
\uput{0.5ex}[d](0.209677,-0.016129){$z^+$}
\uput{0.5ex}[d](0.016129,-0.056452){$o$}
\end{pspicture}%

%% file: pictures/Bapp.tex
\psset{unit=1in,cornersize=absolute,dimen=middle}%
\begin{pspicture}(0,-0.97561)(4,0.432056)%
\psset{linewidth=1pt}%
\newgray{fillval}{0.3}
\pscircle[fillstyle=solid,fillcolor=fillval](0.013937,0){0.013937}
\pscircle[fillstyle=solid,fillcolor=fillval](0.181185,0){0.013937}
\pscircle[fillstyle=solid,fillcolor=fillval](0.348432,0){0.013937}
\pscircle[fillstyle=solid,fillcolor=fillval](0.515679,0){0.013937}
\pscircle[fillstyle=solid,fillcolor=fillval](0.682927,0){0.013937}
\pscircle[fillstyle=solid,fillcolor=fillval](0.850174,0){0.013937}
\pscircle[fillstyle=solid,fillcolor=fillval](1.43554,0){0.013937}
\pscircle[fillstyle=solid,fillcolor=fillval](1.602787,0){0.013937}
\pscircle[fillstyle=solid,fillcolor=fillval](1.770035,0){0.013937}
\pscircle[fillstyle=solid,fillcolor=fillval](1.937282,0){0.013937}
\pscircle[fillstyle=solid,fillcolor=fillval](2.10453,0){0.013937}
\psarcn(0.264808,0.013937){0.083624}{-180}{-360}
\psarcn(0.599303,0.013937){0.083624}{179.999999}{0.000001}
\psarcn(0.432056,0.013937){0.418118}{-180}{-360}
\psarc(1.142857,-0.013937){0.292683}{180}{360}
\psarc(1.142857,-0.013937){0.45993}{180}{360}
\psarc(1.142857,-0.013937){0.627178}{180}{360}
\psarc(1.142857,-0.013937){0.794425}{180}{360}
\psarc(1.142857,-0.013937){0.961672}{180}{360}
\uput{0.5ex}[d](0.850174,-0.013937){$x^+$}
\uput{0.5ex}[d](0.682927,-0.013937){$x^-$}
\uput{0.5ex}[d](0.515679,-0.013937){$y^+$}
\uput{0.5ex}[d](0.348432,-0.013937){$y^-$}
\uput{0.5ex}[d](0.181185,-0.013937){$z^+$}
\uput{0.5ex}[d](0.013937,-0.04878){$o$}
\psframe(1.43554,0.069686)(1.602787,0.209059)
\rput(1.519164,0.139373){$a$}
\psframe(1.770035,0.069686)(1.937282,0.209059)
\rput(1.853659,0.139373){$b$}
\psframe(2.04878,0.069686)(2.160279,0.209059)
\rput(2.10453,0.139373){$c$}
\psframe(3.275261,0.069686)(3.442509,0.209059)
\rput(3.358885,0.139373){$a$}
\psframe(3.581882,0.069686)(3.749129,0.209059)
\rput(3.665505,0.139373){$b$}
\psframe(3.888502,0.069686)(4,0.209059)
\rput(3.944251,0.139373){$c$}
\psline(3.317073,0.069686)(3.317073,-0.069686)
\psarc(3.512195,0.069686){0.111498}{180}{360}
\psarc(3.811847,0.069686){0.10453}{180}{360}
\uput{0.5ex}[d](3.317073,-0.069686){$o$}
\rput(2.71777,0.139373){$=$}
\end{pspicture}%